\documentclass[reprint, aps, pra, superscriptaddress, amsmath, amssymb]{revtex4-2}
\usepackage{physics}
\usepackage[margin=1in]{geometry}
\usepackage{booktabs}
\usepackage{cancel}
\usepackage{xcolor}
\usepackage{mdframed}
\usepackage{verbatim}
\usepackage{amsthm}
\usepackage{graphicx}
\usepackage{float}
\newtheorem{theorem}{Theorem}
\newtheorem*{example*}{Example}
\newtheorem{definition}{Definition}

\newtheorem{lemma}{Lemma}
\newcommand{\Hd}{\mathcal{H}_d}

\newcommand{\CC}{\mathbb{C}}

\newcommand{\Xop}{\mathbf{X}}
\newcommand{\Zop}{\mathbf{Z}}
\newcommand{\Wop}[2]{W(#1,#2)}
\newcommand{\om}{\omega}
\newcommand{\ZZ}{\mathbb{Z}}
\newcommand{\tensn}[1]{{\bigotimes_{i=1}^{n}} #1^{(S_i)}}
\newcommand{\sx}{\sigma_x}
\newcommand{\sy}{\sigma_y}
\newcommand{\sz}{\sigma_z}

\begin{document}

\title{Encrypted Cloning, Absolute Maximal Entanglement and Quantum Secret Sharing}

\author{Zheng Liang Lim}
\email{e0001040@u.nus.edu} 
\affiliation{Centre for Quantum Technologies, National University of Singapore, 3 Science Drive 2, 117543, Singapore}
\affiliation{Department of Physics, National University of Singapore, 2 Science Drive 3, Singapore 117551, Singapore}

\author{Hoi-Kwong Lo}
\thanks{hoikwonglo@nus.edu.sg} 
\affiliation{Centre for Quantum Technologies, National University of Singapore, 3 Science Drive 2, 117543, Singapore}
\affiliation{Department of Physics, National University of Singapore, 2 Science Drive 3, Singapore 117551, Singapore}
\affiliation{Quantum Bridge Technologies Inc., 108 College St., Toronto, ON, Canada}
\affiliation{Department of Electrical and Computer Engineering,
University of Toronto, 10 King’s College Road, Toronto, ON, Canada}

\begin{abstract}
The no-cloning theorem prohibits the creation of identical copies of quantum information, imposing fundamental constraints on quantum technologies. A recently proposed protocol, encrypted cloning, introduced by Yamaguchi and Kempf, showed that perfect qubit clones can be produced if they are simultaneously encrypted with a single-use key. They also observed a connection between this scheme and quantum secret sharing (QSS). However, it remained an open question whether encrypted cloning could be generalised to arbitrary dimensions, and the broader relationship between the two schemes had not been formally established.
In this work, we address both questions by framing encrypted clones as Absolutely Maximally Entangled (AME) states. In parallel with recent work by Ceará that utilises Zadoff-Chu sequences, we independently develop a complementary framework for arbitrary dimensions based on Weyl-Heisenberg displacement operators, both tracing back to the original qubit construction by Yamaguchi and Kempf. We analytically compute the encrypted state and prove that an encrypted qudit system comprising two signal-noise qudit pairs is equivalent to a five-party AME state in any dimension, provided the input state is uniform. We then formalise the connection to QSS by proving that a threshold QSS scheme can achieve the fundamental objectives of encrypted cloning, establishing QSS as the natural general framework within which encrypted cloning can be contextualised.
\end{abstract}

\maketitle

\section{Introduction}
To address the practical bottlenecks created by the no-cloning theorem $\cite{Wootters1982,D.Dieks}$, which prohibits the duplication of quantum information, Yamaguchi and Kempf $\cite{Kempf1}$ proposed a protocol that enables the cloning of encrypted states. This paradigm serves as more than a theoretical curiosity; it naturally enables quantum multicloud storage and offers a promising path toward blind quantum computing by extending the architecture to support parallel homomorphic computation across multiple untrusted clouds~\cite{JF, YF}. Furthermore, the protocol provides enhanced robustness in signal transmission. 

Yamaguchi and Kempf demonstrated that any number of encrypted clones of a qubit can be created via a unitary transformation, and that each clone can subsequently be decrypted through another unitary operation. In their protocol, the system is modelled as an input state accompanied by auxiliary pairs of signal and noise qubits; here, the signal qubits act as potential clones of the input state, while each noise qubit serves as a fixed decryption key. While their work established this architecture for dimension $d=2$, it leaves a research gap regarding the formal algebraic conditions for state recovery in arbitrary dimensions. Furthermore, they observed a relationship between encrypted cloning and quantum secret sharing (QSS); a protocol in which a dealer encodes a secret $S$ in a quantum state that is shared among $n$ players in such a way that only special subsets of players are able to recover the secret~\cite{Cleve, DGottesman}. However, the general framework relating quantum secret sharing (QSS) to encrypted cloning had not been formally contextualised. 

Moreover, at the time of this writing, Ceará has addressed the former open problem by generalising the encrypted cloning protocol to arbitrary dimensions. Notably, he explored Zadoff-Chu sequences for the proposed encryption operators~\cite{Ceara} and analysed their circuit implementations. In this work, we independently derive a distinct framework of unitary operators for encryption and decryption using Weyl-Heisenberg displacement operators. We note that the security requirements for encrypted cloning are inherently more limited than those for a general $((k, N))$ threshold QSS scheme. While the two protocols are closely related, a QSS scheme imposes a stronger requirement. In encrypted cloning, the decryption key is restricted to a fixed set of noise qudits, whereas QSS demands that the secret be recoverable from any valid subset of shares. Therefore, there exist threshold QSS schemes that can be used to achieve the access structure of encrypted cloning, while the converse is not true. Consequently, the noise qudit-based decryption in encrypted cloning represents a specific, constrained case of the more flexible reconstruction process inherent to QSS.

Despite these differences in requirements, we establish that an encrypted system of two signal and noise qudit pairs is mathematically identical to a five-party absolutely maximally entangled (AME) state provided that the input state to be cloned is uniform. AME states, denoted as AME$(n', d)$, are pure $n'$-party states where every reduced density matrix of size $k' \leq \lfloor n'/2 \rfloor$ is maximally mixed~\cite{AMEQSS}. Such states represent the most extremal form of multipartite entanglement, ensuring that any observer holding fewer than half of the parties possesses zero information about the global state. This correspondence holds for arbitrary dimensions, exactly two signal and noise qudit pairs and allows us to reinterpret encrypted cloning through the lens of quantum secret sharing \footnote{The number of signal and noise qudit pairs $n$ has to be more than 1, as when $n=1$, the encrypted clone cannot be decrypted, as shown in the supplementary materials of ~\cite{Kempf1}.}. By leveraging known QSS frameworks, we rigorously formalise the relationship between encrypted cloning and quantum secret sharing, thereby closing the latter research problem. We also show, in passing, how many reductions of the uniform-input
encrypted state and its extension with a reference qudit are maximally mixed for three signal-noise pairs of qudits in any local dimension $d\geq2$. For $d=2$, both states satisfy the slightly weaker perfect extremal maximally entangled (PEME) condition~\cite{Zhang2025}.

The paper is organised as follows. In section II, we review the results given by Yamaguchi and Kempf for encrypted cloning of qubit states and present an overview of the Weyl-Heisenberg operators that will be used to construct our encryption and decryption unitaries. In addition, we also present the construction and application of AME states, followed by an overview of Quantum Secret Sharing and the relationship between AME states and Quantum Secret Sharing. In section III, we propose the encryption and decryption unitary operators for encrypted cloning in arbitrary dimension and examine their validity. Then in section~IV, we demonstrate that an encrypted qudit state for any dimension, comprising two pairs of signal and noise qudits, corresponds
to a five-party AME state. For three pairs of signal and noise qudits, we also show, in passing, that the seven-qudit encrypted state is PEME precisely when $d=2$. In Section~V, we formalise the connection between encrypted cloning
and quantum secret sharing and showed the
existence of a protocol that enables the decryption of
the cloned qudits even if all the decryption keys are
lost.
For the three pairs mentioned above, we also show that the eight-qudit extension with a
reference qudit is PEME for $d=2$. Finally, section VI concludes the paper with a summary and a discussion of future outlooks and directions. The proofs of the statements used in this work are presented in the appendices.

\section{Background}

We begin by reviewing the encrypted cloning protocol for qubits proposed by Yamaguchi and Kempf, before establishing the notation and properties of Weyl-Heisenberg operators. We then demonstrate the connection between states that are maximally entangled across bipartitions and multipartite secret sharing schemes, where information is accessible only through specific authorised subsets.

\subsection{Encrypted Cloning for qubit states}

The original protocol begins with the preparation of $n$ pairs of maximally entangled qubits $(S_i, N_i)$, where $i \in \{1, \dots, n\}$ \cite{Kempf1}. Here, $S_i$ and $N_i$ denote the signal and noise qubits, respectively, each initialised in the Bell state $\ket{\Phi^+}_{S_i N_i} = \frac{1}{\sqrt{2}}(\ket{00} + \ket{11})$. Including the target qubit $\ket{\psi}_A = \alpha\ket{0} + \beta\ket{1}$ to be cloned, the initial state of the full system is given by:
\begin{equation}
\ket{\Psi_0} = \ket{\psi}_A \otimes \bigotimes^n_{i=1} \ket{\Phi^+}_{S_i N_i}.
\end{equation}
Subsequently, the register $A$ interacts with the set of signal qubits $\{S_i\}$ via an encryption unitary $U^{(2,n)}_{enc}$ given by:

\begin{equation}
U^{(2,n)}_{enc}=\frac{1}{2}\sum^3_{\mu=0}\alpha^{-1}_{\mu}\sigma^{(A)}_{\mu}\otimes \bigotimes^n_{i=1}\sigma^{(S_i)}_{\mu}
\end{equation}

such that the noise qubits are not involved in the encryption, as the unitary operator acts as an identity operator in the Hilbert space of $N_i$. Here, $\sigma_0=\mathbb{I}$, $\{\sigma_1,\sigma_2,\sigma_3\}=\{\sigma_x,\sigma_y,\sigma_z\}$ are the set of Pauli matrices, and the $\alpha_\mu$ are normalised coefficients where $|\alpha_{\mu}|^2 = 1$, $\alpha_0=1, \alpha_1=\alpha_3=i, \alpha_2=-(i)^{(n+1)}$. The unitarity of the encryption operator $U^{(2,n)}_{enc}$ is maintained by the phase factors $\alpha_{\mu}$, which are defined such that they facilitate the necessary cancellations in the expansion of $U^{(2,n)}_{enc} U^{(2,n)\dagger}_{enc}$ to give the identity operator. This results in the reduced state of the registers involved in the encryption process i.e. $\rho_A$ and $\rho_{S_i}$ to be maximally mixed. Because a maximally mixed state possesses zero information regarding the basis in which the data was encoded, the data remains perfectly hidden from any external observer who lacks access to the noise qudit required for decryption. After encryption, each signal qubit is regarded as encrypted clone and there are $n > 1$ signal qubits and so there are $n > 1$ encrypted clones. The noise qubits are treated as the decryption keys that are held locally in a quantum memory by the honest party and has to be brought to a signal qudit to unlock the hidden information.

Subsequently, the initial state of qubit $A$ can be recovered from the subset of signal and noise qubits by applying a decryption unitary. The decryption operator to decrypt the $l$-th signal, where $l \in \{1, \dots, n\}$ takes the form

\begin{equation}
U^{(2,n)}_{dec}=\sum^3_{\mu=0} \alpha_{\mu} (\ket{\phi_{\mu}}{\bra{\phi_{\mu}}}_{S_lN_l})\otimes \left(\bigotimes^n_{i\neq l} (\sigma^{(N_i)}_{\mu})^{\mathsf{T}}\right)
\end{equation}
where $\ket{\phi_{\mu}}_{S_lN_l}=(\sigma^{(S_l)}_{\mu}\otimes \mathbb{I} )\ket{\Phi^+}_{S_lN_l}$, $\ket{\phi_0}=\ket{\phi^+}, \ket{\phi_1}=\ket{\psi^+}, \ket{\phi_2}=\ket{\psi^-}, \ket{\phi_3}=\ket{\phi^-}$, which are the Bell states. $\mathsf{T}$ denotes the transpose operation with respect to the computational basis. The unitarity of the decryption operator $U_{dec}^{(2,n)}$ is fundamentally guaranteed by the orthonormality of the Bell basis and the unit norm condition $|\alpha_\mu|^2 = 1$. One can view the operator $\sigma_{\mu}^{\mathsf{T}}$ as being applied conditionally, based on the Bell basis state of the  $S_l, N_l$ subsystem. Following the decryption process, the state of the system $S_l$ becomes the initial state of $A$ and all the noise qubits, are consumed. Consequently, the decryption unitary effectively allows for the complete recovery of the original state using only a single signal qubit, any one of the $n$ encrypted clones, provided it is combined with the held decryption keys. 

It is worth mentioning that, as remarked in ~\cite{Kempf1}, for $d=2$ with an even number of signal-noise pairs $n$, the postencoding register $A$ effectively serves as an encrypted clone. This is because for an even number of signal-noise pairs $n$, the antisymmetric nature of $\sigma_2$ (where $\sigma_2^{\mathsf{T}} = -\sigma_2$) is neutralised within the tensor product. Specifically, the condition $\bigotimes_{i=1}^n \sigma_{\mu}^{(N_i)\mathsf{T}} = \bigotimes_{i=1}^n \sigma_{\mu}^{(N_i)}$ holds for the entire Pauli set $\mu \in \{0, \dots, 3\}$, leaving the encryption unitary invariant and allowing register $A$ to function as an additional clone. This allows for the recovery of the original state from $A$ via a decryption unitary $U^{(2,n)\dagger}_{enc}$, providing a layer of redundancy to the protocol.

It was also remarked that even if one loses some of the noise qubits but retains access to at least one noise qubit, one can still recover the original input state of $A$ from any subsystem composed of one pair of signal and noise qubits plus at least one half of each of the remaining pairs of signal-noise $n-1$ ~\cite{Kempf1}. For example, suppose that we lose $N_2$ and we only have the clones $S_1$ and $S_2$, the surviving system is $S_1S_2N_1N_3\cdots N_n$. One can perform the replacement $\sigma^{(N_2)\mathsf{T}}_{\mu}\rightarrow \sigma^{(S_2)}_{\mu}$ such that the decryption operator becomes:

\begin{equation}
\tilde{U}^{(2,n)}_{dec}= \sum_{\mu=0}^{3} \alpha_{\mu} |\phi_{\mu}\rangle\langle\phi_{\mu}|_{S_1 N_1} \otimes \sigma_{\mu}^{(S_2)} \otimes \bigotimes_{i=3}^{n} \sigma_{\mu}^{(N_i)\mathsf{T}}.
\end{equation}

This works because on the Bell signal-noise pair of system $S_2N_2$: $(\sigma_{\mu}\otimes \mathbb{I})\ket{\Phi^+}_{S_2N_2}=(\mathbb{I}\otimes (\sigma_{\mu})^{\mathsf{T}})\ket{\Phi^+}_{S_2N_2}$. Therefore, acting on the surviving signal qubit $S_2$ plays the same algebraic role as the missing transposed operator on $N_2$. So the recovered state is obtained by applying $\widetilde{U}_{dec}^{(2,n)}$ to the surviving systems and then tracing out the key qubits $N_1, N_3, \dots, N_n$, the signal qubit $S_2$ and system $A$:

\begin{equation}
\operatorname{Tr}_{AS_2N_1 N_3 \dots N_n} \left[ \widetilde{U}_{dec}^{(n)} \rho_{er} \widetilde{U}_{dec}^{(n)\dagger} \right]=\ket{\psi}\bra{\psi}_{S_1},
\end{equation}
$\rho_{er} = \operatorname{Tr}_{N_2} \left[ U_{enc}^{(2,n)} \left( |\psi\rangle\langle\psi|_A \otimes \bigotimes_{i=1}^n |\Phi^+\rangle\langle\Phi^+|_{S_i N_i} \right) U_{enc}^{(2,n)\dagger} \right]$ is the post-erasure state after losing $N_2$.

Hence, the original input state is recovered on $S_1$, even though $N_2$ is missing by sacrificing $S_2$.

\subsection{Weyl-Heisenberg Operators}

Let us denote a $d$-dimensional Hilbert space $\Hd = \mathrm{span}\{\ket{0}, \ket{1}, \ldots, \ket{d-1}\}$, with arithmetic modulo $d$. In this space, the Weyl-Heisenberg operators, which are also known as the generalised Pauli operators, $\Xop$ and $\Zop$ ($\Xop$ is sometimes called the shift operator while $\Zop$ is sometimes called the clock operator) are defined by their action on the computational basis $\{ \ket{k} \}_{k=0}^{d-1}$ as $\Zop\ket{k} = \omega^k\ket{k}$ and $\Xop\ket{k} = \ket{k+1 \pmod d}$, where $\omega = e^{2\pi i/d}$ is the $d$-th root of unity. These operators satisfy the Weyl-Heisenberg commutation relation and periodicity:
\begin{align}
    \Zop^b \Xop^a &= \omega^{ab} \Xop^a \Zop^b, & \Xop^d &= \Zop^d = \mathbb{I},
\end{align}
where $\Xop =\sum_k \ket{k+1}\bra{k}$ and $\Zop =\text{diag}(1, \omega, \dots, \omega^{d-1})$. For $a,b\in\{0,\ldots,d-1\}$, we define the Weyl-Heisenberg displacement operator:
\begin{equation}
  \Wop{a}{b} \;=\; \tau^{ab}\,\Xop^a\Zop^b, \qquad
  \label{eq:Wop}
\end{equation}
where $\tau = e^{\pi i(d+1)/d}$. For odd $d$, changing either label by $d$ leaves the
operator unchanged, whereas for even $d$,
$\Wop{a+d}{b}=(-1)^b\Wop{a}{b}$ and
$\Wop{a}{b+d}=(-1)^a\Wop{a}{b}$. With this convention, the
phase factor $\tau^{ab}$ ensures that
$\Wop{a}{b}^{\dagger}=\Wop{-a}{-b}=\Wop{a}{b}^{-1}$
and $\Wop{a}{b}^{\mathsf{T}}=\Wop{-a}{b}$ hold exactly. The set $\{\Wop{a}{b}\}_{(a,b)\in\ZZ_d^2}$ is an orthonormal basis for the space of all $d \times d$ complex matrices; $M_d(\CC)$ under the Hilbert-Schmidt inner product:
\begin{equation}
  \mathrm{Tr}\!\left[\Wop{a}{b}^\dagger \Wop{a'}{b'}\right] = d\,\delta_{a,a'}\,\delta_{b,b'}
  \label{eq:HS_ortho}
\end{equation}

and follows the composition law given by:
\begin{equation}
  \Wop{a}{b}\,\Wop{a'}{b'} = \tau^{a'b-ab'}\,\Wop{a+a'}{b+b'}.
  \label{eq:composition}
\end{equation}

Leveraging on the properties of the Weyl-Heisenberg operators as discussed above, we arrive at a useful lemma that will help us in our construction of the generalised decryption operator. The proof of the corresponding lemma is shown in Appendix A.

\begin{lemma}
For a maximally entangled state $\ket{\Phi_d} = \frac{1}{\sqrt{d}} \sum_{j=0}^{d-1} \ket{j,j}$, the Weyl-Heisenberg Displacement Operators $\Wop{a}{b} \;=\; \tau^{ab}\,\Xop^a\Zop^b$ satisfy:
\begin{equation}
\begin{aligned}
(\Wop{a}{b} \otimes \mathbb{I})\ket{\Phi_d} &= (\mathbb{I} \otimes \Wop{a}{b}^\mathsf{T}) \ket{\Phi_d} \\
&= (\mathbb{I} \otimes \Wop{-a}{b})\ket{\Phi_d}
\end{aligned}
\end{equation}
where $\omega = e^{2\pi i / d}$ and $\tau=e^{\pi i \frac{(d+1)}{d}}$. Thus, the action of a displacement operator on a single subsystem is locally equivalent to the action of its transpose on the conjugate subsystem.
\end{lemma}

On the side note, for prime $d$, the $d^2-1$ non-identity Weyl-Heisenberg operators partition neatly into $d+1$ subsets of $d-1$ mutually commuting operators, yielding a complete set of Mutually Unbiased Bases. For composite $d$, the index arithmetic modulo $d$ forms a commutative ring rather than a finite field; the resulting zero divisors inherently break this perfect partitioning, although the underlying ring structure of the displacement operators is preserved ~\cite{MUB1,MUB2}.

\subsection{AME states}

Let us review the construction of an AME state in  \cite{AMEQSS}. An absolutely maximally entangled state ~$\text{AME}(n', d)$ is defined as a pure state $|\Psi\rangle \in \mathbb{C}_d^{\otimes n'}$ of $n'$ qudits with dimension $d$ such that for every bipartition of the system into disjoint sets $A'$ and $B'$, the reduced density matrix of the smaller partition is strictly maximally mixed. Formally, for any subset $B'=\{B'_1,\ldots,B'_{m'}\}$ with size $m' = |B'| \leq |A'| = n' - m'$, the state satisfies $\rho_{B'} = \text{Tr}_{A'} |\Psi\rangle\langle\Psi| = \frac{1}{d^{m'}} \mathbb{I}_{d^{m'}}$, where $1\leq m'\leq \frac{n'}{2}$, which implies the von Neumann entropy reaches its theoretical maximum $S(\rho_{B'}) = m' \log_2 d$ \footnote{Here, the indices and labels are primed to avoid confusion with the clashing indices and labels introduced earlier in the previous sections.}. A state is AME if and only if it can be expressed as

\begin{equation}
\ket{\boldsymbol{\Psi}}=\frac{1}{\sqrt{d^{m'}}}\sum_{k\in \ZZ^{m'}_d}\ket{k_1}_{B'_1}\cdots\ket{k_{m'}}_{B'_{m'}}\ket{\Psi(k)}_{A'}
\end{equation}

with $\bra{\Psi(k)}\ket{\Psi(k')}=\delta_{k,k'}$, where equation (11) is obtained from equation (3) of \cite{AMEQSS}. To verify if $\ket{\boldsymbol{\Psi}}$ is AME, it is sufficient to check all $\binom{m'}{\lfloor m'/2\rfloor}$ marginals are maximally mixed, then all smaller marginals are then automatically maximally mixed. 

In the context of protocol of encrypted cloning of qudits, this definition remains physically consistent as the Hilbert space dimension of the partition $A'$ (comprising the register $A$ and signal qudits $S_i$) must be greater than or equal to the dimension of the partition $B'$ (comprising the noise qudits $N_i$), satisfying the dimensionality constraint $d^{|A'|} \geq d^{|B'|}$. 

For qubit systems where dimension $d=2$, AME states exist for $n'=2$ parties, which are the Einstein-Rosen-Pdolsky (EPR) states, and $n'=3$ parties, which are the Greenberger-Horne-Zeilinger (GHZ) states. AME states for $n'=5,6$ qubits were found in ~\cite{Facchi}. However, AME states do not exist for $n'=4$ qubits ~\cite{Higuchi} neither does it exist for $n'\geq 7$ 
qubits~\cite{AScott,FHuber}. 

Let us also introduce a slightly weaker condition. A pure state of
$n'$ qudits of local dimension $d$ is said to be perfect extremal
maximally entangled (PEME) if, for every subsystem size
$1\leq m'\leq\lfloor n'/2\rfloor$, it has the largest possible number
of maximally mixed $m'$-qudit reductions among all pure states with
the same $n'$ and $d$~\cite{Zhang2025}. When an $\mathrm{AME}(n',d)$
state exists, this maximum is $\binom{n'}{m'}$ at every such size,
so the PEME and AME conditions coincide. When no $\mathrm{AME}(n',d)$ state exists, a pure state may still
be PEME if, for each $m'$ from $1$ to $\lfloor n'/2\rfloor$, no other pure state with the same $n'$ and $d$ has more maximally mixed
$m'$-qudit reductions than this state.

\subsection{Quantum Secret Sharing}

 Cleve et al.~\cite{Cleve} and Gottesmann ~\cite{DGottesman} conceptualised the framework of quantum secret sharing (QSS), which is to be distinguished from the sharing of classical secret using quantum resources introduced by Hillery et al.~\cite{Hillery1999}. QSS involves a dealer encoding a secret state $\ket{S} = \sum a_i \ket{i}$ of dimension $d$ into a global state $\sum a_i \ket{\Phi_i}$ shared among $n$ players. The protocol is defined by its ability to restrict information: any subset of players belonging to the access structure (authorised sets) can deterministically recover $\ket{S}$ from their reduced state. Conversely, for any subset in the adversary structure (unauthorised sets), the reduced state is completely independent of the encoded secret. We refer to $\{\ket{\Phi_i}\}$ as the basis states of the QSS scheme. If the total shared state is pure, it constitutes a pure state QSS protocol.

Let us focus our attention on threshold QSS schemes, where the access structure consists of all subsets containing at least $k$ parties while any subset with fewer than $k$ parties belongs to the adversary structure and can obtain no information about the secret. Such a protocol is denoted as a $((k, N))$ threshold scheme, where $N/2 <k\leq N$ such that no-cloning theorem is satisfied. For pure state threshold QSS, the number of players $N$ and the threshold $k$ are strictly related by $N = 2k - 1$. In such a scheme, an authorised set is any subset of participants with size $p$, where $k \leq p \leq N$ and an unauthorised set is any subset of participants with size $q$, where $0 \leq q \leq k-1$. This specific structure allows for a formal connection to absolutely maximally entangled (AME) states:

\begin{theorem}[\cite{AMEQSS}]
There exists a one to one correspondence between an $\text{AME}(2m', d)$ state and a pure state $((m', 2m'-1))$ threshold QSS scheme with secret and share dimensions $d$, characterised by AME basis states.
\end{theorem}

This allows us to formally relate encrypted cloning to QSS via a two-fold approach: first, by demonstrating that the encrypted state is inherently an $\text{AME}(5,d)$ state, as will be shown in section IV, and second, by demonstrating that the partially encrypted Bell state, wherein encryption is applied to only one of the two registers, explicitly realizes an $\text{AME}(6,d)$ state, as detailed in Section V. For completeness, the proof of Theorem 1 is provided in Appendix B.

\section{Operators for Generalised Encrypted Cloning}

To mimic the qubit protocol, the encrypted cloning unitary must imprint the original qudit state across multiple signal qudits, ensure that decryption uses the noise qudits as a key and prevent any single encrypted copy from revealing the data without the key \cite{Kempf1}. This requires extending the unitary construction from combinations of Pauli matrices to combinations of the generalised Weyl-Heisenberg operators $\{\Xop,\Zop\}$ and their tensor products. Meanwhile, the decryption unitary for qudits must act jointly on one signal qudit (the selected copy) and all noise qudits (the key). The structure must ensure that only one decrypted copy can be extracted and the key system is irreversibly consumed during decryption, so that the no-cloning principle is respected. In this section, we attempt to do just that by generalising the encryption and decryption unitaries for the qubit case using the Weyl-Heisenberg displacement operators such that when we set $d=2$ we recover both equation (2) and (3).

We start with a target qudit $A$ in an unknown state $\ket{\psi}_A$ and $n$ pairs of maximally entangled qudits $(S_i, N_i)$, which serve as the signal and noise respectively. Here, the initial state is given by,

\begin{equation}
\ket{\Psi_0}=\ket{\psi}_A\otimes \bigotimes^n_{i=1} \ket{\Phi}_{S_iN_i}
\end{equation}
where $\ket{\Phi}_{S_iN_i}=\frac{1}{\sqrt{d}}\sum^{d-1}_{j_i=0}\ket{j_i,j_i}_{S_iN_i}$ is the $d$ dimensional Bell state, written in the computational basis. Using the Weyl-Heisenberg displacement operators, we can then construct the following encryption unitary $U^{(d,n)}_{enc}$ and decryption unitary $U^{(d,n)}_{dec}$, whose explicit verification of unitarity is provided in Appendix C. They are given by:

\begin{equation}
\begin{split}
U^{(d,n)}_{enc} = \frac{1}{d} \sum_{a,b \in \mathbb{Z}_d} & \varphi^{-1}(a,b) \, W^{(A)}(a,b) \\
& \otimes \bigotimes_{i=1}^{n} \left(W^{(S_i)}(a,b)\right)^{\dagger}, 
\end{split}
\end{equation}
\begin{equation}
\begin{split}
U^{(d,n)}_{dec} = \sum_{a,b \in \mathbb{Z}_d} & \varphi(a,b) \ket{\Psi_{-a,-b}}\bra{\Psi_{-a,-b}}_{S_l N_l} \\
& \otimes \bigotimes_{i \neq l}^{n} \left( W^{(N_i)}(a,b) \right)^{\mathsf{T}}
\end{split}
\end{equation}
where $\ket{\Psi_{-a,-b}}_{S_lN_l}=\tau^{-ab}\bigl(W^{\dag}(a,b)\otimes\mathbb{I})\ket{\Phi}_{S_lN_l}=\tau^{-ab}\bigl(W(-a,-b)\otimes\mathbb{I})\ket{\Phi}_{S_lN_l}$. Here, $\varphi(a,b)\in\CC$ are coefficients to be determined by the unitarity condition, which satisfies $|\varphi(a,b)|^2=1$; a generalisation of $\alpha_{\mu}$. A valid choice for the phase factor is $\varphi(a,b)=\tau^{-(a^2+b^2-(n+1)ab)}$, where $\tau=e^{i\pi (d+1)/d}$ satisfies $\tau^2=\omega$. This specific form is chosen to ensure algebraic consistency with the symmetric Weyl-Heisenberg operators $W(a,b)$. Specifically, the quadratic structure in $a$ and $b$ ensures that the phase contributions from composing displacement operators across all $n+1$ subsystems (1 from register $A$ and $n$ from register $\{S_i\}$) cancel correctly, thereby preserving the unitarity of the global encryption transformation. One can easily see that equation (13) corresponds to the generalisation of (2) and inherently contains the SWAP operator (see Appendix F), while equation (14) is the generalisation of (3), as the Weyl-Heisenberg displacement operator reduces to hermitian Pauli operators when $d=2$. The encryption operator is written such that the operator acts trivially on the noise qudits, while applying identical transformations to both register $A$ and the set of signal qudits. In contrast, the decryption operator is desgined such that it must involve the signal qudit $S_l$ which you are trying to decrypt and the entire collection of the noise qudits $\{N_1,\cdots, N_n\}$.  

To demonstrate that our proposed unitary operators, equations (13) and (14), are consistent with the established qubit results in equations (2) and (3), we evaluate the displacement operators for the $d=2$ case, which yields:
\begin{equation}
\begin{split}
  (\Wop{0}{0})^{\dag}=\mathbb{I}&,\quad
  (\Wop{1}{0})^{\dag}=\Xop^{\dag}=\sigma_x,\quad\\
  (\Wop{0}{1})^{\dag}=\Zop^{\dag}&=\sigma_z,\quad
  (\Wop{1}{1})^{\dag}=\tau^{-1}(\sigma_x\sigma_z)^{\dag} = -\sigma_y.
\end{split}  
\end{equation}
With $\frac{1}{d} = \frac{1}{2}$ and $\varphi(a,b)$ giving $\varphi(0,0)=1$, $\varphi(1,0)=\varphi(0,1)=i$, $\varphi(1,1)=(-i)^{n-1}$, equation (13) gives:

\begin{align}
    U_{enc}^{(2,n)} = \frac{1}{2} \bigg[ &\mathbb{I} \otimes \mathbb{I}^{\otimes n} 
    - i\sigma_x^{(A)} \otimes \left(\sigma_x^{\otimes n}\right)
    - i\sigma_z^{(A)} \otimes \left(\sigma_z^{\otimes n}\right) \nonumber \\
    &+ (i)^{n-1} (-\sigma_y)^{(A)} \otimes \left((-\sigma_y)^{\otimes n}\right) \bigg]
\end{align}

The last term: $(i)^{n-1}(-\sigma_y)^{(A)}\otimes(-\sigma_y)^{\otimes n} = (-1)(-i)^{n+1}\sigma_y^{(A)}\otimes\sigma_y^{\otimes n}$. This reduces to
\begin{equation}
  U_{enc}^{(2,n)} = \frac{1}{2}\sum_{\mu=0}^{3}\alpha_\mu^{-1}\,\sigma_\mu^{(A)}\otimes\tensn{\sigma_\mu},
  \label{eq:d2_result}
\end{equation}
with $\alpha^{-1}_0=1, \alpha^{-1}_1=\alpha^{-1}_3=-i$ and $\alpha^{-1}_2=-(-i)^{n+1}$, matching equation (2). To obtain the expression of equation (14) when $d=2$, we use $\Wop{a}{b}^{\mathsf{T}}=\Wop{-a}{b}$ and the $d=2$ values of the Bell
states $\ket{\Psi_{-a,-b}}=\tau^{-ab}(W(-a,-b)\otimes\mathbb{I})\ket{\Phi}$:
\begin{alignat}{2}
  \Wop{0}{0}^{\mathsf{T}} &= \mathbb{I},&\quad
    \ket{\Psi_{0,0}} &= \ket{\phi^+},\notag\\
  \Wop{1}{0}^{\mathsf{T}} &= \sx,&\quad
    \ket{\Psi_{1,0}} &= (\sx\otimes\mathbb{I})\ket{\Phi}
      = \ket{\psi^+},\notag\\
  \Wop{0}{1}^{\mathsf{T}} &= \sz,&\quad
    \ket{\Psi_{0,1}} &= (\sz\otimes\mathbb{I})\ket{\Phi}
      = \ket{\phi^-},\notag\\
  \Wop{1}{1}^{\mathsf{T}} &= -\sy^{\mathsf{T}} = \sy,&\,\
    \ket{\Psi_{1,1}} &= (i)(-\sy\otimes\mathbb{I})\ket{\Phi}\notag\\
      & & &= \ket{\psi^-},
\end{alignat}
where $\ket{\phi^+} = \ket{\phi_0}, \ket{\psi^+} = \ket{\phi_1}, \ket{\phi^-} = \ket{\phi_3}$ and $\ket{\psi^-} = \ket{\phi_2}$. For the tensor factors, noting $\Wop{1}{1}^{\mathsf{T}}=-\sigma_2^{\mathsf{T}}$
so each of the $(n-1)$ copies contributes $(-1)$:
\begin{equation}
  \bigotimes_{i\neq l}^{n}\bigl(\Wop{a}{b}^{(N_i)}\bigr)^{\!\mathsf{T}}
  =
  \begin{cases}
    \bigotimes_{i=2}^{n}\sigma_\mu^{(N_i)\mathsf{T}},\; {\scriptstyle(a,b)\neq(1,1)},\\[4pt]
    (-1)^{n-1}\bigotimes_{i\neq l}^{n}\sigma_2^{(N_i)\mathsf{T}}, \,{\scriptstyle (a,b)=(1,1)}.
  \end{cases}
\end{equation}

Substituting the values of $\varphi(a,b)$ as stated above and noting that $\varphi(1,1)(-1)^{n-1}= (-i)^{(n-1)}(-1)^{(n-1)} = (i)^{n-1}=-(i)^{n+1}$,
\begin{equation}
\begin{split}
  U_{dec}^{(2,n)}
    = \sum_{\mu=0}^{3}\alpha_\mu\,
      \ket{\phi_\mu}\!\bra{\phi_\mu}_{S_1N_1}
      \otimes
      \bigotimes_{i\neq l}^{n}\sigma_\mu^{(N_i)\mathsf{T}}
\end{split}     
\end{equation}
with $\alpha_0=1$, $\alpha_1=\alpha_3=i$, $\alpha_2=-(i)^{n+1}$,
matching equation (3). 

In section IIA, we established that the encrypted state obtained after applying the encryption operator must satisfy two conditions. First, for any single-register subset $R \in {A,S_1,\dots,S_n}$, the reduced density operator $\rho_R$ must be maximally mixed in order to conceal the encoded information from an external observer i.e. $\rho_A=\Tr_{S_1,S_2,\cdots,S_n,N_1,\cdots N_n}(U^{(d,n)}_{enc}\ket{\Psi_0}\bra{\Psi_0}U^{(d,n)\dag}_{enc})=\mathbb{I}_d/d$ and $\rho_{S_i}=\Tr_{A,S_1,\cdots,S_{i-1},S_{i+1},\cdots,S_n,N_1,\cdots N_n}$  $   (U^{(d,n)}_{enc}\ket{\Psi_0}\bra{\Psi_0}U^{(d,n)\dag}_{enc})=\mathbb{I}_d/d$ for $i\in\{1,2,\cdots,n\}$. Second, the action of the decryption operator on the encrypted state must recover the original state in register $A$ from the subsystem $S_lN_1\cdots N_n$ for any $l\in{1,\dots,n}$. In this picture, the noise qudits ${N_i}$ act as a decryption key: together they decouple a chosen signal qudit $S_l$, which then carries the original state $\ket{\psi}$ from register $A$ i.e. $\rho_{S_l}=\Tr_{A,S_1,\cdots,S_{l-1},S_{l+1},\cdots,S_n,N_1,\cdots,N_n}(U^{(d,n)}_{dec}U^{(d,n)}_{enc}\ket{\Psi_0}$ $\bra{\Psi_0}U^{(d,n)\dag}_{enc}U^{(d,n)\dag}_{dec})=\ket{\psi}\bra{\psi}_{S_l}$. The remaining signal qudits remain entangled with their corresponding noise partners and therefore cannot be accessed as independent pure states. Consequently, only one signal qudit can be ``unlocked'' at a time, ensuring that only a single pure copy of $\ket{\psi}$ is recoverable and thus preserving consistency with the no-cloning theorem. These properties certify the validity of the proposed encryption and decryption operators for the protocol. For completeness, detailed proofs are deferred to Appendix D.

Also, note that unlike $d=2$, the postencoding qudit $A$ cannot serve as an encrypted clone regardless of whether $n$ is even or odd for $d>2$ as the Weyl-Heisenberg operator $W(a,b)$ is symmetric i.e. $W(a,b)^{\mathsf{T}}=W(a,b)$ if $a=0$ or ($d$ is even, $a=d/2$, $b$ even), while it is anti-symmetric i.e. $W(a,b)^{\mathsf{T}}=-W(a,b)$  if and only if $d$ is even, $a=d/2$, and $b$ is odd. For the property $W(a,b)^{\mathsf{T}} = \pm W(a,b)$ to hold for every element in the basis, which is the requirement for qubit $A$ to serve as an encrypted clone, the condition $2a \equiv 0 \pmod d$ must be satisfied for all $a \in \{0, 1, \dots, d-1\}$. This can only happen when $d=2$.

However, like $d=2$, even if we lose some of the  $n$ noise qudits, one can still recover the original state $A$ from any subsystem composed of one pair of signal and noise qubits plus at least one half of each of the remaining $n-1$ pairs of signal-noise qudit pairs. An example where $N_2$ is lost but the original state can still be recovered from $S_1$, by replacing it with $S_2$, is shown in Appendix E below, where the derivation is slightly lengthier than that for $d=2$.

\section{Relationship between encrypted states and AME states}

In section III, we have focused on encrypted cloning as a dynamic process to distribute quantum information of the state of register $A$. It is worthwhile to note that it also provides a recipe for generating highly entangled states. By initialising the input register in a uniform superposition, the protocol produces a global encrypted state $|\Psi_{enc}\rangle$ that acts as a multipartite entanglement resource. In fact, for arbitrary dimensions, this construction yields an Absolutely Maximally 
Entangled (AME) state. 

\begin{definition}
Let $A$ be a $d$-dimensional quantum register. The initial state of the register, denoted by $|\psi\rangle_A$, is defined as the uniform superposition state$$|\psi\rangle_A = \frac{1}{\sqrt{d}} \sum_{j=0}^{d-1} |j\rangle,$$where $\{|j\rangle\}_{j=0}^{d-1}$ denotes the standard computational basis of the Hilbert space $\mathbb{C}^d$.
\end{definition}

With that, we formalise the relationship in the following theorem:

\begin{theorem}
Let $d \geq 2$ and $n \geq 2$. Consider a system of $n$ encrypted signal-noise pair qudits together with an input register $A$ initialized in the uniform superposition state $\ket{\psi}_A$ defined in definition 1. The resulting global encrypted state $\ket{\Psi_{enc}} \in (\mathbb{C}^d)^{\otimes (2n+1)}$ is an absolutely maximally entangled state $\mathrm{AME}(2n+1, d)$ if and only if $n = 2$. For $n > 2$, $\ket{\Psi_{enc}}$ fails to satisfy the property that all reduced density matrices $\rho_{E}$ are maximally mixed for subsets $|E| = n$.
\end{theorem}

This means that for all subsets of qudits $E \subset \{A,S_1,N_1,S_2,N_2\}$ with $|E|\leq 2$, the marginal density matrix is maximally mixed: $\rho_{E}=\frac{1}{d^{|E|}}\mathbb{I}_{d^{|E|}}$. To verify that a state $|\Psi\rangle$ of $m'$ qudits is Absolutely Maximally Entangled, it is sufficient to show that every marginal of size $\lfloor m'/2 \rfloor$ is maximally mixed. Specifically, one checks all $\binom{m'}{\lfloor m'/2 \rfloor}$ such marginals; if these are maximally mixed, all smaller marginals are guaranteed to be so. For the $m'=5$ case considered here, this requires demonstrating that all $\binom{5}{2}=10$ bipartite marginals satisfy $\rho_E = \frac{1}{d^2} \mathbb{I}_{d^2}$. Direct evaluation of these 10 marginals is straightforward but notationally tedious. By writing the $n=2$ state explicitly as

\begin{equation}
\begin{split}
    |\Psi_{enc}\rangle&=U^{(d,2)}_{enc}\ket{\Psi_0}\\
    &= \frac{1}{d^{2}}\sum_{a,b}\tau^{a^2+b^2}
      \sum_{k,j_1,j_2}\psi_k \omega^{b(k-j_1-j_2)}\\
      &|k+a,j_1-a,j_2-a\rangle_{AS_1S_2}
      \otimes|j_1,j_2\rangle_{N_1N_2},
\end{split}      
\end{equation}
one can exploit the permutation symmetries of the system to partition the verification into three distinct groups:

\begin{table}[h]
\centering
\begin{tabular}{@{}lll@{}}
\toprule
\textbf{Group} & \textbf{Marginals} & \textbf{Type} \\ \midrule
Group 1 & $\rho_{S_1N_1}, \rho_{S_2N_2}$ & Within-pair \\
Group 2 & $\rho_{AS_1}, \rho_{AN_1}, \rho_{AS_2}, \rho_{AN_2}$ & $A$ + 1 noise or signal \\
Group 3 & $\rho_{S_1S_2}, \rho_{S_1N_2}, \rho_{N_1S_2}, \rho_{N_1N_2}$ & Between pairs \\ \bottomrule
\end{tabular}
\end{table}

While this group by group partial trace calculation rigorously confirms the AME property, it is far too exhaustive. However, for completeness, we relegate this lengthy, explicit computation to Appendix G. While in Appendix H, we show that at least one of the three-register marginal density operators is not maximally mixed.
Here, we provide a sketch of the proof of Theorem 2. 

\begin{proof}[Sketch of proof]
We treat the cases $n=2$ and $n>2$ separately.

For $n=2$, Appendix G evaluates all $10$ two-qudit marginals of
$\ket{\Psi_{enc}}$ and shows that
\begin{equation}
\rho_E=\frac{\mathbb{I}_{d^2}}{d^2}
\qquad
\text{for every }
E\subseteq\{A,S_1,N_1,S_2,N_2\},
\quad |E|=2.
\end{equation}
Hence $\ket{\Psi_{enc}}$ is an $\mathrm{AME}(5,d)$ state.

Suppose now that $n>2$. To disprove the AME property, it is enough
to find one $n$-qudit marginal that is not maximally mixed. We
therefore choose an $n$-qudit subsystem $E$ for which no initial
signal-noise Bell pair is split by the bipartition $E|E^c$, where $E^c$ is the complement of $E$. When
$n$ is even, let $E$ contain the first $n/2$ complete pairs
$S_iN_i$. When $n$ is odd, let $E$ contain the input qudit $A$
together with the first $(n-1)/2$ complete pairs. In either case,
$E$ contains exactly $n$ qudits, and every initial Bell pair lies
entirely on one side of the cut.

The encryption unitary is a linear combination of $d^2$
tensor-product Weyl operators. Since the initial state is a product
state across the particular cut chosen above, each Weyl summand
produces a product vector across that cut. Consequently,
$\ket{\Psi_{enc}}$ is a sum of at most $d^2$ product
vectors across $E|E^c$, and hence
\begin{equation}
\operatorname{SR}_{E\mid E^c}
\bigl(\ket{\Psi_{enc}}\bigr)\leq d^2.
\end{equation}

Because the global encrypted state is pure, the rank of its reduced
state on $E$ equals the Schmidt rank across the same cut. Therefore,
\begin{equation}
\operatorname{rank}(\rho_E)
=
\operatorname{SR}_{E\mid E^c}
\bigl(\ket{\Psi_{enc}}\bigr)
\leq d^2
<
d^n,
\end{equation}
where the strict inequality holds for $d\geq2$ and $n>2$. Thus
$\rho_E$ is not maximally mixed. Hence
$\ket{\Psi_{enc}}$ is not an
$\mathrm{AME}(2n+1,d)$ state for $n>2$.
\end{proof}

In this vein, Theorem 2 implies that the encrypted states of $n=2, d=2$ correspond to AME$(5,2)$ state, while the encrypted states of $n=2, d=3$ correspond to AME$(5,3)$ state, so on and so forth. This is of interest, as encrypted cloning has a larger implication than it seems. In ~\cite{AJScott, SBall}, it was shown that $|\Psi\rangle$ is an absolutely maximally entangled state if and only if the subspace spanned by $|\Psi\rangle$ forms a pure $\{[(d_1, \ldots, d_{n'}), 1, \lfloor m'\rfloor + 1]\}$ quantum error-correcting code, where $d_i$ are the local dimensions of each subsystem $i$ \footnote{It was shown in ~\cite{AMEQSS} that a 5-qubit code can be used to construct an $\text{AME}(5,2)$ state, though the authors did not establish the converse.}. In other words, a state is AME if and only if it can be viewed as a single codeword of a quantum error-correcting code with maximum possible distance $\lfloor m'\rfloor + 1$, meaning it can correct errors on any subset of up to $\lfloor m'/2\rfloor$ subsystems. This means that an AME$(5,2)$ state exists if and only if the
[[5,1,3]] perfect 5-qubit code exists, where $d_1=d_2=\cdots=d_{n'}=2$. The two are thus equivalent objects. This means that the encrypted states can be used as an error correction code.

On a side note, for three signal-noise pairs of qubits, one can say precisely how the encrypted state departs from the AME condition. Set $n=3$ and take the uniform input
$\ket{\psi}_A=d^{-1/2}\sum_{j=0}^{d-1}\ket{j}_A$ of Theorem 2.
For every $d\geq2$, the encrypted state has 32 maximally mixed three-qudit reductions out of 35. The only exceptions are the three
subsystems $\{A,S_i,N_i\}$, $i=1,2,3$, each of which has rank
$d^2$ (see Appendix~H) \footnote{These are the authorised coalitions that can decrypt the input. In the original work by Yamaguchi and Kempf, the user will remove the post-encoded register $A$ to prevent unauthorised decryption.}. For $d=2$, at most 32 of the 35 three-qubit reductions of a pure
seven qubit state can be maximally mixed~\cite{FHuber,Zhang2025}, so the encrypted
state is a perfect extremal maximally entangled (PEME) state~\cite{Zhang2025}. For every $d\geq3$, an $\mathrm{AME}(7,d)$ state exists and has all 35 three qudit reductions maximally mixed~\cite{FShi}. Since every AME state satisfies the PEME definition, the uniform-input encrypted state has only
32 maximally mixed three-qudit reductions, therefore falls
short of the PEME criterion for $d\geq3$.

\section{Mapping Encrypted Cloning to Quantum Secret Sharing}

The general correspondence between AME states and
pure-state threshold QSS schemes was established by
Helwig et al.~\cite{AMEQSS} and is recalled in
Theorem~1. Here, we give an explicit realisation of
this correspondence using our encrypted cloning
construction with $n=2$ signal-noise pairs. We verify
the reduced-state conditions for the resulting
five-qudit encoding and show that its coherent
extension yields an AME$(6,d)$ state, thereby
providing a pure-state $((3,5))$ threshold QSS scheme.

\begin{theorem}
For a $d$-dimensional Bell state under partial encryption, where only a single register is encrypted, there exist an encrypted cloning protocol with $n=2$ signal-noise qudit pairs (up to a local unitary) that exhibits a one to one correspondence to a pure state $((3,5))$ threshold QSS scheme with secret and share dimensions $d$, where $d\geq 2$.
\end{theorem}

We provide a sketch of the proof Theorem 3 as a coherent extension of Theorem 2.

\begin{proof}[Sketch of proof]
Let
\begin{equation}
V:\mathcal H_d\longrightarrow\mathcal H_d^{\otimes5}
\end{equation}
denote the encrypted cloning isometry
\begin{equation}
V\ket{\psi}_A
=U_{enc}^{(d,2)}
\left(
\ket{\psi}_A\otimes
\ket{\Phi}_{S_1N_1}\otimes
\ket{\Phi}_{S_2N_2}
\right).
\label{eq:encoding_isometry}
\end{equation}

Consider the Fourier basis
\begin{equation}
\ket{\chi_r}
=\frac{1}{\sqrt d}\sum_{k\in\mathbb Z_d}\omega^{rk}\ket{k}=\sum_k \psi_k\ket{k},
\qquad r\in\mathbb Z_d,
\end{equation}
and define
\begin{equation}
\ket{\Psi_r}=V\ket{\chi_r}.
\end{equation}
The state $\ket{\chi_0}$ is the uniform input used in Theorem 2. Since
$\ket{\chi_r}=\Zop^r\ket{\chi_0}$ where $\Zop$ is the generalised Pauli operator as mentioned earlier, Weyl-covariance of the encryption unitary gives, up to an irrelevant scalar phase,
\begin{equation}
\begin{split}
&U_{enc}^{(d,2)}(\Zop_A^r\otimes \mathbb{I}_{S_1}\otimes \mathbb{I}_{S_2})U_{enc}^{(d,2)\dagger}\\
&=
\left[
W^{(A)}(1,1)W^{(S_1)}(-1,0)W^{(S_2)}(-1,0)
\right]^r
\label{eq:Fourier_covariance}
\end{split}
\end{equation}
, as demonstrated in Appendix I. The operator on the right-hand side is a tensor product of local unitaries. Consequently, every $\ket{\Psi_r}$ is locally unitarily equivalent to $\ket{\Psi_{enc}}$. By Theorem 2, every $\ket{\Psi_r}$ is therefore an $\mathrm{AME}(5,d)$ state.

Let $E$ be any two-party subset of the five encrypted systems and $E^c$ be its complement. Theorem 2 gives the diagonal part of the required marginal relation. To extend it, one can retain the terms with different Fourier labels and trace out the complement of $E$. For every two-party subsystem
$E\subset\{A,S_1,N_1,S_2,N_2\}$, 
\begin{equation}
\operatorname{Tr}_{E^c}
\left(
\ket{\Psi_r}\!\bra{\Psi_s}
\right)
=\delta_{r,s}\frac{\mathbb I_E}{d^2},
\qquad |E|=2.
\label{eq:coherent_AME5_condition}
\end{equation}
For $r=s$, equation (30) is essentially the maximally mixed two-party marginal established in Theorem 2. To illustrate the off-diagonal calculation where $r\neq s$, take
$E=\{A,S_1\}$ and $E^c=\{S_2,N_1,N_2\}$.
Substituting the Fourier amplitudes into equation (21) (where unprimed and primed indices label independent
summations in $|\Psi_r\rangle$ and $\langle\Psi_s|$,
respectively),
the partial trace imposes $j_1'=j_1$, $j_2'=j_2$,
and $a'=a$. The successive sums over $j_2$ and $b$
then enforce $b'=b$ and $k'=k$, respectively,
leaving
$\operatorname{Tr}_{E^c}(|\Psi_r\rangle\langle\Psi_s|)
=d^{-3}\sum_{a,k,j_1\in\mathbb Z_d}
\omega^{(r-s)k}
|k+a,j_1-a\rangle\langle k+a,j_1-a|$.
For each fixed $k$, the map
$(a,j_1)\mapsto(k+a,j_1-a)$ is bijective, so this
expression becomes
$d^{-3}\bigl(\sum_{k\in\mathbb Z_d}
\omega^{(r-s)k}\bigr)\mathbb{I}_E$.
The remaining sum is evaluated using the
root of unity identity
\begin{equation}
\sum_{k\in\mathbb Z_d}\omega^{(r-s)k}
=d\,\delta_{r,s},
\end{equation}
so the off-diagonal reduced operator vanishes. 

We now consider a $d$-dimensional Bell state
$\ket{\Phi}_{A_1A_2}=d^{-1/2}\sum_{j=0}^{d-1}\ket{jj}_{A_1A_2}$ and apply the encryption unitary to $A_2,S_1,S_2$:
\begin{equation}
\label{eq:u_enc}
\begin{split}
U'^{(d,2)}_{\text{enc}} &= \frac{1}{d} \sum_{a,b \in \mathbb{Z}_d} \varphi^{-1}(a,b) \,  W^{(A_2)}(a,b) \\
&\quad \otimes \bigotimes_{i=1}^{2} \left(W^{(S_i)}(a,b)\right)^{\dag}, 
\end{split}
\end{equation}
which acts exclusively on register $A_2$ and the signal qudits $\{S_1,S_2\}$. This yields the partially encrypted state
\begin{equation}
\begin{split}
\label{eq:psi_enc}
&|\Psi'_{\text{enc}}\rangle \\
&= \left(\mathbb{I}_{A_1} \otimes U'^{(d,2)}_{\text{enc}} \otimes \mathbb{I}_{N_1N_2}\right)\left(|\Phi\rangle_{A_1A_2} \otimes \bigotimes_{i=1}^{2}|\Phi\rangle_{S_iN_i}\right).
\end{split}
\end{equation}

Writing the bell state $\ket{\Phi}_{A_1A_2}$ as a Fourier-basis Schmidt decomposition yields
\begin{equation}
\ket{\Phi}_{A_1A_2}
=\frac{1}{\sqrt d}\sum_{r\in\mathbb Z_d}
\ket{\chi_r'}_{A_1}\ket{\chi_r}_{A_2},
\end{equation}
we may equivalently write
\begin{equation}
\ket{\Psi'_{enc}}
=\frac{1}{\sqrt d}\sum_{r\in\mathbb Z_d}
\ket{\chi_r'}_{A_1}\ket{\Psi_r}_{A_2S_1N_1S_2N_2}.
\label{eq:coherent_AME_extension}
\end{equation}

\begin{figure}[h]
    \centering
    \includegraphics[width=0.5\textwidth]{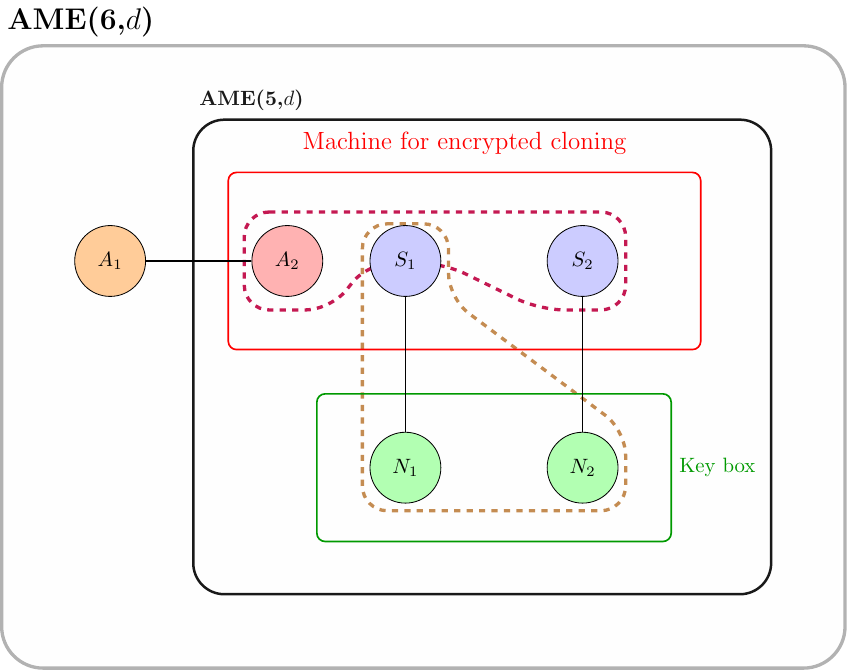}
    \caption{A schematic diagram of the partial encrypted cloning protocol for $n=2,d=2$. The machine contains the signal qubits and register $A_2$, on which the encryption operator acts, while the key box contains the noise qubits. Conditioning the reference register $A_1$ on a Fourier-basis outcome prepares one of the five-party AME states $\ket{\Psi_r}$. Their coherent extension with $A_1$ forms the $\mathrm{AME}(6,2)$ state. The brown dotted lines show an example of an authorised set $\{S_1,N_1,N_2\}$, while the purple dotted lines show an example of an unauthorised set $\{A_2,S_2\}$, illustrating a $((3,5))$ threshold QSS scheme.}
    \label{fig:a}
\end{figure}

 Using equation (30), the reduced state on $A_1E$, where $A_1E$ is a three-party subset of the six partially encrypted system, is
\begin{equation}
\begin{split}
\rho_{A_1E}
&=\frac{1}{d}\sum_{r,s\in\mathbb Z_d}
\ket{\chi_r'}\!\bra{\chi_s'}_{A_1}
\otimes
\operatorname{Tr}_{E^c}
\left(\ket{\Psi_r}\!\bra{\Psi_s}\right)\vert_{r=s}\\
&=\frac{1}{d}\sum_{r\in\mathbb Z_d}
\ket{\chi_r'}\!\bra{\chi_r'}_{A_1}
\otimes\frac{\mathbb I_E}{d^2}\\
&=\frac{\mathbb I_{A_1E}}{d^3}.
\end{split}
\end{equation}
Hence every three-party marginal containing $A_1$ is maximally mixed. Since $\ket{\Psi'_{enc}}$ is pure, the reduced state on the complementary three-party subsystem has the same nonzero spectrum and is also equal to $\mathbb I_{d^3}/d^3$. Every three-party subset either contains $A_1$ or is complementary to one that contains $A_1$. Therefore every three-party marginal is maximally mixed, and $\ket{\Psi'_{enc}}$ is an $\mathrm{AME}(6,d)$ state.

Finally, we invoke Theorem 1, which dictates a one to one correspondence (up to local unitaries) between an $\text{AME}(2m',d)$ state and a pure-state $((m', 2m'-1))$ threshold QSS scheme. Setting $m'=3$, our constructed $\text{AME}(6,d)$ state uniquely maps to a $((3,5))$ threshold QSS scheme with secret and share dimensions $d$, where $d\geq 2$.
\end{proof}

To further support our argument, let us provide an example where an AME$(6,2)$ is constructed from AME$(5,2)$ states. Note that the construction in the Example can be adapted to arbitrary dimension:

\begin{example*}
In this example, we show how the encrypted state, which can be used to construct an AME$(5,2)$ as seen above, can be used to construct an AME$(6,2)$ state. Expanding the encrypted state, given by equation (21), for $d=2$ gives:

\begin{equation}
\begin{split}
\lvert \Psi_{enc} \rangle
= \frac{1-i}{4\sqrt{2}} \Big[ 
& \lvert 00000 \rangle 
+ \lvert 01111 \rangle 
+ \lvert 10101 \rangle 
+ \lvert 11010 \rangle \\
+ &\lvert 11001 \rangle 
+ \lvert 10110 \rangle 
+ \lvert 01100 \rangle 
+ \lvert 00011 \rangle 
\Big] \\
+ \frac{1+i}{4\sqrt{2}} \Big[
& \lvert 00101 \rangle 
+ \lvert 01010 \rangle 
+ \lvert 10000 \rangle 
+ \lvert 11111 \rangle 
\Big] \\
+ \frac{-1-i}{4\sqrt{2}} \Big[
& \lvert 11100 \rangle 
+ \lvert 10011 \rangle 
+ \lvert 01001 \rangle 
+ \lvert 00110 \rangle 
\Big]
\end{split}
\end{equation}

Recall from Theorem 2, $\ket{\Psi_{enc}}$ given by equation (21) is an AME$(5,2)$ state. From the result of ~\cite{AJScott, SBall}, an AME(5,2) state is equivalent to a single codeword of a pure $[[5,1,3]]$ code. Since $|\Psi_{enc}\rangle$ is an AME(5,2) state, it serves as a valid ``logical zero" $|0'_L\rangle_{AS_1S_2N_1N_2}$ for a non-standard code basis where:

\begin{equation}
\begin{split}
\ket{0'_L}_{AS_1S_2N_1N_2}&=\ket{\Psi_{enc}}\\
&=U^{(2,n)}_{enc}(\ket{+}\otimes\ket{\Phi^+}\otimes\ket{\Phi^+}.
\end{split}
\end{equation}

We can therefore construct a $\ket{1'_L}_{AS_1S_2N_1N_2}$ from:

\begin{equation}
\begin{split}
\ket{1'_L}&_{AS_1S_2N_1N_2}=U^{(2,n)}_{enc}(\ket{-}\otimes\ket{\Phi^+}\otimes\ket{\Phi^+}\\
= &\frac{1-i}{4\sqrt{2}} \Big[ |00000\rangle + |01111\rangle + |11001\rangle + |10110\rangle \Big]\\
+ &\frac{1+i}{4\sqrt{2}} \Big[|00101\rangle + |00110\rangle + |01001\rangle + |01010\rangle\Big] \\
+ & \frac{-1-i}{4\sqrt{2}}\Big[|10000\rangle + |10011\rangle + |11100\rangle + |11111\rangle\big] \\
- & \frac{1-i}{4\sqrt{2}} \Big[|00011\rangle + |01100\rangle + |10101\rangle + |11010\rangle\Big].
\end{split}
\end{equation}

Note that this is not a unique construction. One can easily check that $\bra{0'_L}\ket{0'_L}=\bra{1'_L}\ket{1'_L}=1$ and using unitarity of $U^{(2,n)}_{enc}$, $\bra{0'_L}\ket{1'_L}=0$. One can also utilise the receipe shown in Appendix G to show that $\ket{1'_L}$ is also an AME$(5,2)$ state. Then, using the receipe of Theorem 1, we can obtain the AME$(6,2)$ state using:

\begin{equation}
\ket{AME(6,2)}=\frac{1}{\sqrt{2}}(\ket{0}\ket{0'_L}+\ket{1}\ket{1'_L})\\
\end{equation}

Since $\ket{0_L'}$ and $\ket{1'_L}$ form an orthonormal pair of AME$(5,2)$ states, which span the codeword space of a $[[5,1,3]]$ ~\cite{AJScott, SBall}, this enables the construction of AME$(6,2)$, which one can verify computationally. In fact, equation (40) related to equation (33) via a local unitary transformation on the register $A$.

\end{example*}

In the Example, we have shown how to construct an AME$(6,2)$ state using an orthogonal pair of encrypted states, which are by themselves AME$(5,2)$ states. Therefore, by Theorem 1, there exist a one to one correspondence between a partially encrypted state with $n=2$ signal and noise qubits and a pure state $((3,5))$ threshold QSS scheme. We can apply a similar construction to an arbitrary dimension by starting with the AME$(5,d)$ given in equation (21), following Theorem 2 and thereafter, following the recipe of Theorem 1 to construct an AME$(6,d)$ state and relate it to the QSS ((3,5)) threshold scheme.

The relationship established between the $n=2$ partially encrypted state and the $((3,5))$ threshold QSS scheme, as illustrated in Figure 1, provides a physical intuition for the ``sharing" of quantum information. In the context of encrypted cloning, the machine does not merely copy a state; it distributes the information of the input state across a multipartite entangled system such that the information is hidden from local parties but remains globally accessible.

On a side note, the same construction extends to three signal-noise pairs in any
local dimension $d\geq2$. Keeping the reference qudit $A_1$
gives an eight-qudit state. For $d=2$, all reductions on at most three
qubits and exactly 56 of the 70 four-qubit reductions are maximally
mixed. This achieves the eight-qubit PEME condition~\cite{Zhang2025}
(see Appendix~H for arbitrary $d$). Note that in dimensions where an $\mathrm{AME}(8,d)$ state exists,
the PEME condition requires all 70 four-qudit reductions to be
maximally mixed. The encrypted state falls short of this condition
in those dimensions, since not all of its reductions of this size
are maximally mixed. These reductions also motivate
non-threshold QSS. The reduction on $A_1N_1N_2N_3$ is maximally mixed,
so $\{N_1,N_2,N_3\}$ is private. The reduction on $A_1A_2S_1N_1$ is
one of the exceptions, and the equally sized coalition $\{A_2,S_1,N_1\}$
is qualified because a generalised Bell measurement on $S_1N_1$
identifies the Weyl correction that recovers the secret on $A_2$.

\section{Conclusion}

In this manuscript, we have shown that the encrypted qubit cloning protocol can be generalised to qudits of arbitrary dimension using the Weyl-Heisenberg operator and the proposed encryption and decryption operators are shown to be consistent to those introduced by Yamaguchi and Kempf at dimension $d=2$. We realised that, like $d=2$, even if some of the noise qudits are lost, one can still recover the original input state by sacrificing some of the signal qudits, albeit with a slightly different decryption unitary operator. We have also studied the concept of AME states and QSS and demonstrated that the an encrypted state comprising of two pairs of signal and noise qudits is equivalent to a five-party AME state when the input state to be cloned is uniform.

Furthermore, we investigated the relationship between the encrypted state and QSS scheme. We proved that for arbitrary dimension $d\geq 2$, the encrypted state is an AME$(5,d)$ state for $n=2$ provided the input state to be cloned is uniform. On the other hand, for $n > 2$, the process does not result in an AME$(2n+1, d)$ state, with both results established in Theorem 2. To highlight the connection to QSS, we proved that when the input is a $d$-dimensional Bell state subjected to partial encryption, where only one of the registers is encrypted, there exists a one to one correspondence between the $n=2$ encrypted cloning protocol and a pure-state $((3,5))$ threshold QSS scheme for arbitrary $d \geq 2$, as established in Theorem 3. 

In the future work, we would like to explore the usage of encrypted cloning as an error correction code and examine the relationship between stabilzer codes and encrypted states ~\cite{Pawel}. Encrypted Cloning of qudits also provides a natural route towards blind quantum computing (BQC), where Weyl-masked qudits (qudits that are encrypted by the Weyl-Heisenberg operator) act as encrypted inputs transmitted to a remote server. The server performs computations directly on the masked data, while only the client, who retains the noise register keys, can resolve the final phases and recover the output. Thereafter, one can consider using the encrypted state for $n=2$ to correct the errors from the computation. 

Another direction for future work is to characterise the entanglement
of encrypted states beyond the AME case. For the uniform input in
Theorem~2, the encrypted state is AME precisely when $n=2$. For $n=3$ and every local dimension $d\geq2$, we have determined how
many reductions of the seven qudit encrypted state and its eight qudit extension (with an additional reference qudit) are maximally mixed at each size up to half the system. For $d=2$, both states satisfy the slightly weaker perfect extremal maximally entangled (PEME) condition~\cite{Zhang2025}. For $d\geq3$, the seven qudit uniform-input state falls short of this condition. The eight-qudit extension also falls short whenever an $\mathrm{AME}(8,d)$ state exists. This motivates extending the analysis
to $n>3$ signal-noise pairs to determine when encrypted cloning
gives rise to PEME states. Furthermore, recent work
by Liu and Shao~\cite{LiuShao2026} examines non-threshold qubit QSS
schemes and the relationship between entanglement properties and
access structures. Their work provides a broader context for our
QSS $((3,5))$ threshold construction and motivates future
investigation of possible connections between encrypted cloning
and non-threshold QSS schemes.

On the implementation side, simulating these high dimensional encrypted cloning on advanced superconducting platforms such as the IBM Heron processor will be essential \cite{Kempf2}. Studying performance under hardware specific noise profiles will allow one to identify the gate error thresholds, thereby establishing practical limits for secure, high-dimensional blind quantum computation. 

\section*{Acknowledgement}

We thank Koji Yamaguchi and Achim Kempf and the anonymous referee for the insightful discussion. We would also like to thank Xing Jian Zhang for proof reading the draft. This project is supported by the National Research Foundation, Singapore through the National Quantum Office, hosted in A*STAR, under its Centre for Quantum Technologies Funding Initiative. Hoi-Kwong Lo acknowledges research support from the
National University of Singapore Start-up grant, the National Research Council of Canada High Throughout Secure Networks (HTSN) program, Canadian Foundaton for Innovation (CFI) and the Natural Sciences and Engineering Research Council of Canada (NSERC) Discovery Grant.

\section*{Appendix}

\subsection*{Appendix A: Proof of Lemma 1}

\begin{proof}
Let $M = \sum_{i,k} M_{i,k} |i\rangle\langle k|$ be an arbitrary linear operator. Acting on the first subsystem of the maximally entangled state, we have:
\begin{align*}
    (M \otimes \mathbb{I}) |\Phi_d\rangle &= \frac{1}{\sqrt{d}} \sum_j M |j\rangle \otimes |j\rangle \\
    &= \frac{1}{\sqrt{d}} \sum_j \left( \sum_{i,k} M_{i,k} |i\rangle\langle k| \right) |j\rangle \otimes |j\rangle \\
    &= \frac{1}{\sqrt{d}} \sum_{i,j} M_{i,j} |i\rangle \otimes |j\rangle \label{eq:basis_expansion}
\end{align*}
Using the definition of the transpose, $(M^{\mathsf{T}})_{j,i} = M_{i,j}$. We associate the coefficients with the second subsystem:
\begin{align*}
    (M \otimes \mathbb{I}) |\Phi_d\rangle &= \frac{1}{\sqrt{d}} \sum_{i,j} |i\rangle \otimes M_{i,j} |j\rangle \\
    &= \frac{1}{\sqrt{d}} \sum_i |i\rangle \otimes \left( \sum_j M^{\mathsf{T}}_{j,i} |j\rangle \right) \\
    &= (I \otimes M^{\mathsf{T}}) |\Phi_d\rangle
\end{align*}

We now evaluate $W(a,b)^{\mathsf{T}}$. Recall that $W(a,b) = \tau^{ab} \Xop^a \Zop^b$. Given the properties $\Zop^{\mathsf{T}} = \Zop$ and $\Xop^{\mathsf{T}} = \Xop^{-1}$, and noting that the transpose of a product reverses the order, we find:
\begin{equation*}
    W(a,b)^{\mathsf{T}} = (\tau^{ab} \Xop^a \Zop^b)^{\mathsf{T}} = \tau^{ab} (\Zop^b)^{\mathsf{T}} (\Xop^a)^{\mathsf{T}} = \tau^{ab} \Zop^b \Xop^{-a}
\end{equation*}
Since $\Zop^b \Xop^a = \omega^{ab} \Xop^a \Zop^b$, this implies $\Zop^b \Xop^{-a} = \omega^{-ab} \Xop^{-a} \Zop^b$. Substituting this in:
\begin{equation*}
    W(a,b)^{\mathsf{T}} = \tau^{ab} \omega^{-ab} \Xop^{-a} \Zop^b
\end{equation*}
Recalling that $\omega = \tau^2$ (where $\tau = e^{i\pi(d+1)/d}$ and $\omega=e^{i 2\pi/d}$), we have $\omega^{-ab} = \tau^{-2ab}$:
\begin{equation*}
    W(a,b)^{\mathsf{T}} = \tau^{ab} \tau^{-2ab} \Xop^{-a} \Zop^b = \tau^{-ab} \Xop^{-a} \Zop^b = W(-a,b)
\end{equation*}
Combining the two stages yields:
\begin{equation*}
\begin{split}
    (W(a,b) \otimes \mathbb{I}) |\Phi_d\rangle = (\mathbb{I}\otimes W(a,b)^{\mathsf{T}}) |\Phi_d\rangle\\
    =(\mathbb{I} \otimes W(-a,b)) |\Phi_d\rangle
\end{split}    
\end{equation*}
\end{proof}
\subsection*{Appendix B: Proof of Theorem 1}

\begin{proof}
AME$\implies$ QSS: Let $|\boldsymbol{\Psi}\rangle \in \mathcal{H}_d^{\otimes 2m'}$ be an AME state. We partition the $2m'$ parties into two sets of size $m'$: $B' = \{D', B'_1, \dots, B'_{m'-1}\}$ and $A' = \{A'_1, \dots, A'_{m'}\}$. The state can be Schmidt-decomposed across the bipartition as:
\begin{equation*}
    |\boldsymbol{\Psi}\rangle = \frac{1}{\sqrt{d^{m'}}} \sum_{i \in \mathbb{Z}_d} \sum_{k \in \mathbb{Z}_d^{m'-1}} |i\rangle_D |k_1 \dots k_{m'-1}\rangle_{B'} |\Psi(i, k)\rangle_{A'}
\end{equation*}
Since $\rho_{B'} = \text{Tr}_{A'}(|\boldsymbol{\Psi}\rangle\langle\boldsymbol{\Psi}|) = \frac{\mathbb{I}}{d^{m'}}$, the set $\{|\Psi(i, k)\rangle_{A'}\}$ constitutes an orthonormal basis for $\mathcal{H}_{A'}$, satisfying $\langle \Psi(i, k) | \Psi(j, k') \rangle = \delta_{ij}\delta_{kk'}$.

To encode a secret $|a\rangle = \sum a_i |i\rangle$, the Dealer projects their system $D'$ onto the basis $\{|i\rangle\}$, yielding the QSS encoded state:
\begin{equation*}
\begin{split}
    |a\rangle = &\sum_i a_i |\boldsymbol{\Psi}_i\rangle,\,\ \\
    &\text{where } |\boldsymbol{\Psi}_i\rangle = \frac{1}{\sqrt{d^{m'-1}}} \sum_k |k\rangle_{B'} |\Psi(i, k)\rangle_{A'}
\end{split}    
\end{equation*}

For any $m'-1$ parties in $B' \setminus \{D'\}$, the reduced density matrix is:
\begin{equation*}
    \rho_{B' \setminus \{D'\}} = \text{Tr}_{A'} (|a\rangle\langle a|) = \frac{1}{d^{m'-1}} \sum_k |k\rangle\langle k| = \frac{\mathbb{I}}{d^{m'-1}}
\end{equation*}
This is independent of $a_i$, ensuring zero information leakage to unauthorised parties.

Since $\{|\Psi(i, k)\rangle_{A'}\}$ is orthonormal, there exists a unitary $U_{A'}$ such that $U_{A'} |\Psi(i, k)\rangle_{A'} = |k_1\rangle_{A'_1}\dots|k_{m'-1}\rangle_{A'_{m'-1}} |i\rangle_{A'_{m'}}$. Applying $U_{A'}$ to the shares in $A'$ recovers the secret:
\begin{equation*}
\begin{split}
    (\mathbb{I}_{B'} \otimes U_{A'}) |a\rangle &= \left( \frac{1}{\sqrt{d^{m'-1}}} \sum_k |k\rangle_{B'} |k\rangle_{A' \setminus \{A'_{m'}\}} \right)\\
    &\otimes \sum_i a_i |i\rangle_{A'_{m'}}
\end{split}    
\end{equation*}
The secret is recovered on party $A'_{m'}$.

QSS $\implies$ AME:
Consider a QSS scheme that encodes a secret $i$ into states $|\boldsymbol{\Psi}_i\rangle$ across $2m'-1$ parties. We define a global state $|\boldsymbol{\Psi}\rangle$ by entangling an auxiliary system $D'$ (the Dealer):
\begin{equation}
    |\boldsymbol{\Psi}\rangle = \frac{1}{\sqrt{d}} \sum_{i \in \mathbb{Z}_d} |i\rangle_D |\boldsymbol{\Psi}_i\rangle
\end{equation}
Expanding $|\boldsymbol{\Psi}_i\rangle$ in terms of the unauthorised shares $k$ and authorised shares $A'$:
\begin{equation*}
    |\boldsymbol{\Psi}\rangle = \frac{1}{\sqrt{d^{m'}}} \sum_{i, k} |i\rangle_{D'} |k\rangle_{B'} |\Psi(i, k)\rangle_{A'}
\end{equation*}
To recover $i$ deterministically from $A'$, we must have $\langle \phi(k, i) | \phi(k, j) \rangle = \delta_{ij}$. Also, to ensure $B'$ gains no information, $\rho_B$ must be maximally mixed, requiring $\langle \phi(k, i) | \phi(k', i) \rangle = \delta_{kk'}$.

Thus, $\langle \phi(i, k) | \phi(j, k') \rangle = \delta_{ij}\delta_{kk'}$. This implies that for the bipartition $B' \cup \{D'\}$ vs $A'$, the state is maximally entangled. Since the choice of authorised/unauthorised sets in a threshold scheme is arbitrary for any set of size $m'$, $|\boldsymbol{\Psi}\rangle$ is maximally entangled with respect to any bipartition into two cardinality $m'$ sets, satisfying the definition of an $AME(2m', d)$ state.

\end{proof}

\subsection*{Appendix C: Proof of Unitarity of generalised encryption and decryption operator}
\textbf{Proof of unitarity of generalised encryption operator:}
\begin{proof}
Let
$\widetilde{W}(a,b) \equiv W^{(A)}(a,b)\otimes\bigotimes_{i=1}^n (W^{(S_i)}(a,b))^{\dag}$
denote the full tensor product Weyl operator over all $n+1$ subsystems.

Therefore,
\begin{equation*}
\begin{split}
    U_{enc}^{(d,n)}U_{enc}^{(d,n)\dagger}
    &= \frac{1}{d^2}
    \sum_{a,b,a',b'\in\mathbb{Z}_d}\\
    &\varphi^{-1}(a,b)\,\varphi(a',b')\;
    \widetilde{W}(a,b)\,\widetilde{W}^\dagger(a',b').
\end{split}    
\end{equation*}

Using $W^\dagger(a',b')=W(-a',-b')$ and the Weyl composition rule
\begin{equation*}
\begin{split}
    W(a,b)\,W(a',b') = \tau^{a'b - ab'}\,W(a+a',b+b'),\\
    W^{\dag}(a,b)\,W^{\dag}(a',b')=\tau^{a'b - ab'}\,W^{\dag}(a+a',b+b'),
\end{split}    
\end{equation*}

applied to each of the $n+1$ tensor factors:
\begin{equation*}
    \widetilde{W}(a,b)\,\widetilde{W}^\dagger(a',b')
    = \tau^{(n+1)(ab'-a'b)}\,\widetilde{W}(a-a',\,b-b').
\end{equation*}

Substituting $a'=a-m,\; b'=b-p$:
\begin{equation*}
    U_{enc}^{(d,n)}U_{enc}^{(d,n)\dagger}
    = \frac{1}{d^2}
    \sum_{m,p\in\mathbb{Z}_d}
    S(m,p)\;
    \widetilde{W}(m,p),
\end{equation*}
where
\begin{equation*}
    S(m,p)
    \equiv
    \sum_{a,b\in\mathbb{Z}_d}
    \varphi^{-1}(a,b)\,\varphi(a-m,\,b-p)\;
    \tau^{(n+1)(mb-ap)}.
\end{equation*}

Since $\varphi(a,b)=\tau^{-(a^2+b^2-(n+1)ab)}$:
\begin{align*}
    \varphi^{-1}(a,b)&\,\varphi(a-m,b-p)
    = \tau^{a^2+b^2-(n+1)ab}
       \cdot\\
       &\tau^{-\left[(a-m)^2+(b-p)^2-(n+1)(a-m)(b-p)\right]} \notag\\
    &= \tau^{2am - m^2 + 2bp - p^2 - (n+1)(ap + bm - mp)}.
\end{align*}

Hence, $S(m,p)$ factorises as:
\begin{equation*}
\begin{split}
    S(m,p)
    =&
    \tau^{-m^2-p^2+(n+1)mp}
    \left(\sum_{a\in\mathbb{Z}_d}\tau^{2a(m-(n+1)p)}\right)\\
    &\times \left(\sum_{b\in\mathbb{Z}_d}\tau^{2bp}\right).
\end{split}    
\end{equation*}

Since $\tau^2 = \omega$, each sum is a standard sum over the root of unity where,
\begin{equation*}
    \sum_{b\in\mathbb{Z}_d}\tau^{2bp}
    = \sum_{b\in\mathbb{Z}_d}\omega^{bp}
    = d\,\delta_{p,0}.
\end{equation*}

which enforces $p = 0$, and the sum over $a$ reduces to:
\begin{equation*}
    \left.\sum_{a\in\mathbb{Z}_d}\tau^{2a(m-(n+1)p)}\right|_{p=0}
    = \sum_{a\in\mathbb{Z}_d}\omega^{am}
    = d\,\delta_{m,0}.
\end{equation*}

Therefore:
\begin{equation*}
    S(m,p) = d^2\,\delta_{m,0}\,\delta_{p,0}.
\end{equation*}
Which implies that
\begin{equation*}
\begin{split}
    U_{enc}^{(d,n)}U_{enc}^{(d,n)\dagger}
    &= \frac{1}{d^2}
    \sum_{m,p\in\mathbb{Z}_d}
    d^2\,\delta_{m,0}\,\delta_{p,0}\;
    \widetilde{W}(m,p)\\
    &= \widetilde{W}(0,0)
    = \mathbb{I}.
\end{split}    
\end{equation*}

Since $W^{(A)}(0,0)\otimes\bigotimes_i (W^{(S_i)}(0,0))^{\dag} = \mathbb{I}$, we conclude:
\begin{equation*}
U_{enc}^{(d,n)}\,U_{enc}^{(d,n)\dagger} = \mathbb{I}.
\end{equation*}
The same computation for $U_{enc}^{(d,n)\dagger} U_{enc}^{(d,n)}$, swapping the roles
of $(a,b)$ and $(a',b')$, proceeds identically, confirming that the generalised encryption operator is unitary.

\end{proof}

\textbf{Proof of unitarity of generalised decryption operator:}

\begin{proof}
Let the  the Bell state projector for $S_lN_l$: $\ket{\Psi_{-a,-b}}\!\bra{\Psi_{-a,-b}}_{S_lN_l}=P_{a,b}$,

Therefore,
\begin{equation*}
\begin{split}
  U^{(d,n)}_{dec}&U^{(d,n)\dag}_{dec}
  = \sum_{a,b}\sum_{a',b'}
     \varphi(a,b)\,\varphi^{*}(a',b')\;
     (P_{a,b}\,P_{a',b'})_{S_lN_l}\\
     &\;\otimes\;
     \bigotimes_{i\neq l}^{n}
     \Bigl[
       W(a,b)^{\mathsf{T}}
       \bigl(W(a',b')^{\mathsf{T}}\bigr)^{\!\dagger}
     \Bigr]^{(N_j)}.
     \label{eq:expand}
\end{split}     
\end{equation*}

where $P_{a,b}\,P_{a',b'} = \delta_{a,a'}\,\delta_{b,b'}\,P_{a,b}$ which enforces $a=a'$ and $b=b'$,

Since the Weyl-Heisenberg displacement operator $W$ is unitary:
\begin{equation*}
\begin{split}
  W^{\mathsf{T}}\bigl(W^{\mathsf{T}}\bigr)^{\!\dagger}
  &= W^{\mathsf{T}}\bigl(W^{\dagger}\bigr)^{\mathsf{T}}
  = W^{\mathsf{T}}\bigl(W^{-1}\bigr)^{\mathsf{T}}\\
  &= \bigl(W^{-1}W\bigr)^{\mathsf{T}}
  = \mathbb{I}.
\end{split}  
\end{equation*}

\noindent Therefore,
\begin{align*}
  U^{(d,n)}_{dec}U^{(d,n)\dag}_{dec}
  &= \sum_{a,b} |\varphi(a,b)|^2\;P_{a,b}
     \;\otimes\; \mathbb{I}^{\otimes(n-1)}
  \notag\\[4pt]
  &= \left(\,\sum_{a,b\in\ZZ_d} \ket{\Psi_{-a,-b}}\!\bra{\Psi_{-a,-b}}\right)
     \otimes\, \mathbb{I}^{\otimes(n-1)},
  \label{eq:collapsed}
\end{align*}
where in the second line we used $|\varphi(a,b)|^2=1$.
The $d^2$ generalised Bell states form a complete orthonormal basis of
$\mathcal{H}^{(S_l)}\otimes\mathcal{H}^{(N_l)}$, so:
\begin{equation*}
  \sum_{a,b\in\ZZ_d} \ket{\Psi_{-a,-b}}\!\bra{\Psi_{-a,-b}} = \mathbb{I}.
  \label{eq:completeness}
\end{equation*}
Therefore
\begin{equation*}
  U_{dec}^{(d,n)}\,\bigl(U_{dec}^{(d,n)}\bigr)^{\!\dagger}
  = \mathbb{I} \otimes \mathbb{I}^{\otimes(n-1)}
  = \mathbb{I}. 
\end{equation*}
again, one can easily show that $U^{(d,n)\dag}_{dec}U^{(d,n)}_{dec}=\mathbb{I}$ as well using the same derivation.

\end{proof}

\subsection*{Appendix D: Validity of unitaries}

Consider the initial state $\ket{\Psi_{0}}$ given by equation (12) where we let $\ket{\psi}_A=\sum_k \psi_k\ket{k}_A$. The encryption operator acts as
$U_{enc}^{(d,n)} = \frac{1}{d}\sum_{a,b\in \ZZ_d}\varphi^{-1}(a,b)\,
W^{(A)}(a,b)\otimes\bigotimes_{i=1}^{n}\left(W^{(S_i)}(a,b)\right)^{\dagger}$,
applying $W(a,b)$ to $A$ and each $S_i$, leaving the $N_i$ untouched.
We consider the action on $\ket{\psi}_A$:
\begin{equation*}
\begin{split}
  W^{(A)}(a,b)\ket{\psi}_A
  = \tau^{ab}\sum_k \psi_k\, \Xop^a \Zop^b\ket{k}\\
  = \tau^{ab}\sum_k \psi_k\,\om^{bk}\ket{k+a}.
\end{split}  
\end{equation*}

We then consider the action on each $\ket{\Phi}_{S_i N_i}$.
$\left(W^{(S_i)}(a,b)\right)^{\dag}$ acts only on $S_i$:
\begin{equation*}
\begin{split}
  \left(W^{(S_i)}(a,b)\right)^{\dag}\ket{\Phi}_{S_i N_i}
  &= \frac{\tau^{ab}}{\sqrt{d}}\sum_{j_i}\om^{-b j_i}\ket{j_i-a}_{S_i}\ket{j_i}_{N_i}\\
  &= \tau^{ab}\ket{\Psi_{-a,-b}}_{S_i N_i},
\end{split}  
\end{equation*}

Writing them together and expanding the Bell States yields $\ket{\Psi_{enc}}=U_{enc}^{(d,n)}\ket{\Psi_0}$ as:
\begin{equation*}
\begin{split}
  &\frac{1}{d}\sum_{a,b}\varphi^{-1}(a,b)\,\tau^{ab}
     \Bigl(\sum_k\psi_k\,\om^{bk}\ket{k+a}_A\Bigr)\\
     &\otimes\bigotimes_{i=1}^{n}\tau^{ab}\ket{\Psi_{-a,-b}}_{S_i N_i}\\
  &= \frac{1}{d}\sum_{a,b}\varphi^{-1}(a,b)\,\tau^{(n+1)ab}
     \sum_k\psi_k\,\om^{bk}\ket{k+a}_A\\
     &\otimes\bigotimes_{i=1}^{n}\ket{\Psi_{-a,-b}}_{S_i N_i}
\end{split}
\end{equation*}

which is also equal to:
\begin{equation*}
\begin{split}
  &\frac{1}{d}\cdot\frac{1}{d^{n/2}}
     \sum_{a,b}\varphi^{-1}(a,b)\,\tau^{(n+1)ab}
     \sum_k\psi_k\\
   & \times \sum_{j_1,\ldots,j_n}
     \om^{b(k-j_1-\cdots-j_n)}
     \ket{k+a}_A
     \otimes\bigotimes_{i=1}^{n}\ket{j_i-a, j_i}_{S_i N_i}.
\end{split}
\end{equation*}
Defining $\mathbf{j} = (j_1,\ldots,j_n)$ and
$|\mathbf{j}| = j_1+\cdots+j_n$,
\begin{equation*}
\begin{split}
  \ket{\Psi_{enc}}
  &= \frac{1}{d^{(n+2)/2}}
    \sum_{a,b}\varphi^{-1}(a,b)\,\tau^{(n+1)ab}\\
    &\sum_{k,\mathbf{j}}
    \psi_{k}\,\om^{b(k-|\mathbf{j}|)}
    \ket{k+a}_A
    \otimes\bigotimes_{i=1}^{n}\ket{j_i-a,\, j_i}_{S_i N_i}.
  \label{eq:Psi_enc}
\end{split} 
\end{equation*}

The sum over $b$ acts as a $\mathbb{Z}_d$ Fourier transform:
$\sum_b \om^{b(k-|\mathbf{j}|)} = d\,\delta_{k-|\mathbf{j}|,\,0\bmod d}$ as it is a geometric series sum of the root of unity. For the encryption to remain secure, $\varphi(a,b)$ must be chosen such that the sum enforces a correlation between $k$ and the collective indices $|\mathbf{j}|$ only when the correct decryption operator is applied. In the absence of this operator, the original information $|\psi\rangle_A$ is delocalized across the $n$-party system. Consequently, tracing out other subset of qudits, leaving only $A$, yields a maximally mixed state. To see this, we compute $\rho_{A} = \text{Tr}_{S_1,\,S_2\cdots S_n,\,N_1\cdots N_n}(\ket{\Psi_{enc}}\!\bra{\Psi_{enc}})$.

Let the primed variables $(a',b',k',\mathbf{j}')$ label the bra. The partial trace
imposes orthogonality conditions from each traced-out register:
\begin{align*}
  &\text{Trace over }S_iN_i:\\
   & \bra{j_i'-a',j_i'}\ket{j_i-a,j_i} = \delta_{j_i-a,\,j_i'-a'}\,\delta_{j_i,j_i'}.\label{eq:trSiNi}
\end{align*}
The trace condition above gives $j_i = j_i'$ and $a = a'$.
Hence $a$, $b$, $b'$, $k$, $k'$ and $j_i$ are the remaining free indices. Since the sum over the $n$ free indices $j_1, j_2,\ldots,j_n$
contributes $d^{n}$:
\begin{equation*}
  \sum_{j_1,j_2,\ldots,j_n}\om^{(b'-b)|\mathbf{j}|}
  = \prod_{i=1}^{n}\sum_{j_i=0}^{d-1}\om^{(b'-b)j_i}
  = d^{n}\,\delta_{b,b'}.
  \label{eq:ji_sum}
\end{equation*}

Hence, the reduced 
density matrix $\rho_A$, obtained by tracing over $\ket{\Psi_{enc}}$ becomes (after summing over $b'$),
\begin{equation*}
\begin{split}
  \rho_{A}
  &= \frac{1}{d^{2}}
     \sum_{a,b,k, k'}
     \varphi^{-1}(a,b)\varphi^{-1*}(a,b)\,
     \tau^{(n+1)ab}\tau^{-(n+1)ab}\,\\
     &\psi_k\psi^{*}_{k'}\,
     \om^{b(k-k')}\,
     \ket{k+a}\!\bra{k'+a}_{A}.
\end{split}     
\end{equation*}

Since $\sum_k|\psi_k|^2=1$ and the phases $\om^{b(k-k')}$
sum over $b$ to $d\,\delta_{k,k'}$, which enforces $k=k'$ and noting $|\varphi^{-1}(a,b)|^2=1$:

\begin{align*}
  \rho_{A}
  &= \frac{1}{d}
     \sum_k |\psi_k|^2\sum_{a}
     \,\ket{k+a}\!\bra{k+a}_{A} \notag\\
  &=\frac{\mathbb{I}_d}{d}.
\end{align*}

The reduced state on $A$ is the maximally mixed (white noise) state. Now we consider the reduced state of $S_1$. Here, we trace over $S_iN_i$ where $i\neq 1$ and over $A$ giving $a=a'$ =, $j_i=j'_i$ for $i\neq 1$ and $k=k'$ sum over $n-1$ free indices $j_2,j_3\cdots,j_n$ which contributes $d^{n-1}$:

\begin{equation*}
\begin{split}
\sum_{j_2,j_3,\cdots,j_n}\omega^{(b'-b)\sum^n_{i=2}j_i} =\prod^{n}_{i=2}\sum^{d-1}_{j_i=0}\omega^{(b'-b)j_i}=d^{n-1}\delta_{b,b'}
\end{split}
\end{equation*}

Hence, the reduced density matrix $\rho_{S_1}$, obtained by tracing over $\ket{\Psi_{enc}}$ becomes:

\begin{equation*}
\begin{split}
\rho_{S_1}=& \frac{1}{d^3}\sum_{k,a,b,j_1,j'_1}\varphi^{-1}(a,b)\varphi^{-1*}(a,b)\,
     \tau^{(n+1)ab}\tau^{-(n+1)ab}\,\\
     &|\psi_k|^2\omega^{b(j'_1-j_1)}\,
     \ket{j_1-a}\!\bra{j'_1-a}_{S_1}\\
     &=\frac{1}{d^2}\sum_{a,j_1,j'_1}\delta_{j_1,j'_1}\ket{j_1-a}\!\bra{j'_1-a}_{S_1}\\
     &=\frac{1}{d^2}\sum_{a,j_1}\ket{j_1-a}\!\bra{j_1-a}_{S_1}=\frac{\mathbb{I}_d}{d}.
\end{split}
\end{equation*}
Since $a, j_1\in\mathbb{Z}_d$, $\rho_{S_1}$ is also maximally mixed. Hence, due to the symmetry of the exchange of the roles of $S_i$ for all $i\in \{1,\cdots, n\}$, both $\rho_A$ and $\rho_{S_i}$ are maximally mixed, and the data remains perfectly hidden from any external observer.

We now apply the decryption unitary to decrypt the $l^{th}$ signal qudit. The operator must involve $S_l$ and the entire collection of the noise qudits $\{N_1,...,N_n\}$. Let $j'_l, a', b'$ be the indices for the $S_l, N_l$ component of the encrypted state and $j_l, a, b$ be that for the projector operator $\ket{\Psi_{-a,-b}}\bra{\Psi_{-a,-b}}_{S_lN_l}$. Then, applying the projector to the signal and noise component of $\ket{\Psi_{enc}}$ yields:

\begin{equation*}
\begin{split}
\bra{\Psi_{-a,-b}}_{S_lN_l}\left(\frac{1}{\sqrt{d}}\sum_{j'_l} \omega^{-b'.j'_l}\ket{j'_l - a'}_{S_l}\ket{j'_l}_{N_l}\right)\\
=\frac{1}{d}\sum_{j_l, j'_l\in \ZZ_d}\omega^{-b'j'_l + bj_l}\underbrace{\bra{j_l-a}\ket{j'_l-a'}}_{
\delta_{j_l-a, j'_l-a'}
}\underbrace{\bra{j_l}\ket{j'_l}}_{\delta_{j_l,j'_l}}\\
=\delta_{a,a'}\sum_{j_l}\omega^{j_l(b-b')}\times \frac{1}{d}\\
=\delta_{a,a'}\delta_{b,b'}
\end{split}
\end{equation*}

 One can view the operator $(W^{(N_i)}(a,b))^{\mathsf{T}}$ as being applied conditionally, based on the Bell basis state of the  $S_l, N_l$ subsystem. Note that $\sum_{j_l}\omega^{j_l(b-b')}$ is a geometric series sum of the root of unity. If $b\neq b'$, the term vanishes and if $b=b'$, $\omega^0=1$, which sums to $d$ (cancelling the $1/d$).

Since $(W^{(N_i)}(a,b))^{\mathsf{T}}=W^{(N_i)}(-a,b)$, we can write the output state of $U^{(d,n)}_{dec}\ket{\Psi_{enc}}$ as,

\begin{equation*}
\begin{split}
\ket{\Psi_{out}}&=\frac{1}{d}\sum_{a,b}\tau^{(n+1)ab}\sum_{k}\psi_k\,\omega^{bk}\ket{k+a}_A\otimes \\
&\ket{\Psi_{-a,-b}}_{S_lN_l}\otimes \bigotimes^n_{i\neq l}W^{(N_i)}(-a,b)\ket{\Psi_{-a,-b}}_{S_iN_i}
\end{split}
\end{equation*}

We now combine this correction with the $j$-th signal qudit $S_j$ from the encrypted state. The state of the $j$-th pair becomes (you have $n-1$ copies of them):
\begin{equation*}
\begin{split}
    W^{(N_j)}&(-a,b)\ket{\Psi_{-a,-b}}_{S_jN_j}\\
    &= \frac{\tau^{-ab}}{\sqrt{d}}\sum_{j_j} \ket{j_j-a}_{S_j}\otimes \ket{j_j-a}_{N_j}  
\end{split}    
\end{equation*}

with that, $\ket{\Psi_{out}}$ becomes

\begin{equation*}
\begin{split}
\ket{\Psi_{out}}&=\frac{1}{d^{(n+2)/2}}\sum_{a,b}\omega^{ab}\sum_{k,\mathbf{j}}\psi_k\,\omega^{b(k-j_l)}\ket{k+a}_A\otimes\\
&\ket{j_l-a}_{S_l}\otimes\ket{j_l}_{N_l}\otimes \bigotimes^n_{i\neq l}\ket{j_i-a, j_i-a}_{S_iN_i}.
\end{split}
\end{equation*}
By summing over $b$ and noting that $\sum_b \omega^{b(k-j_l+a)}=d\delta_{m,0}$ where $m=k-j_l+a=0\implies a=j_l-k$ (mod $d$), $\ket{\Psi_{out}}$ simplifies to:

\begin{equation*}
\begin{split}
\ket{\Psi_{out}}=&\frac{1}{d^{n/2}}\sum_{k, \mathbf{j}}\ket{j_l}_A\otimes \psi_k\ket{k}_{S_l}\otimes \ket{j_l}_{N_l}\\
&\otimes \bigotimes^n_{i\neq l}\ket{j_i-j_l+k, j_i-j_l+k}_{S_iN_i}.
\end{split}
\end{equation*}

The equation above provides a hint that the original input state in register $A$ can be recovered from $S_l$. This is possible because the encryption operator we have proposed in equation (13), in fact, contains the SWAP operator (see Appendix F). 

Since $\rho_{S_l}=\mathrm{Tr}_{A,S_1,..S_{l-1},S_{l+1},..S_n,N_1,..N_n}\left(\ket{\Phi_{out}}\bra{\Phi_{out}}\right)$, by tracing out the other subsystems, we have

\begin{flalign*}
  &\text{Trace over }A, N_l: \bra{j_l'}\ket{j_l} = \delta_{j_l,j_l'},\\
  &\text{Trace over }S_iN_i\ (i\neq l):\notag\\
  & \begin{aligned}[b]
    &\bra{j_i'-j'_l+k',j_i'-j'_l+k'}\ket{j_i-j_l+k,j_i-j_l+k} \\
    &= \delta_{j_i-j_l+k,j_i'-j'_l+k'}.
  \end{aligned}
\end{flalign*}

The trace condition over $S_iN_i$ gives $j_i+k=j'_i+k'$ for all $i \neq l$.
While trace condition over $A, N_l$ gives $j_l = j_l'$. Moreover, since the sum over the $(n-1)$ free indices $j_1,\ldots j_{l-1},j_{l+1},\ldots,j_n$ (which we will now call $j_i\{i\neq l\}$) contributes a factor of $d^{n-1}$.

 After working out the terms with $i\neq l$ and noting that the only remaining free index alongside are $k,k',j_l$ and $j'_l$, the density matrix becomes
\begin{equation*}
\begin{split}
  \rho_{S_l} &= \frac{1}{d} \sum_{k,k',j_l,j'_l} \psi_k \psi^*_{k'} \cdot \delta_{j_l,j'_l}|k\rangle\langle k'|
\end{split}     
\end{equation*}

Therefore, we recover:

\begin{equation*}
\rho_{S_l}=\sum_{k,k'}\psi_k\psi^*_{k'}\ket{k}\bra{k'}=\ket{\psi}\bra{\psi}_{S_l}
\end{equation*}

which holds for any initial pure state $\ket{\psi}_A$ and any $l\in\{1,\cdots, n\}$. 

\subsection*{Appendix E: Recovery of input state despite losing noise qudits}

Let the Weyl operators and phases be defined as
\[
\omega=e^{2\pi i/d},
\qquad
\tau=e^{\pi i(d+1)/d},
\qquad
W(a,b)=\tau^{ab}\mathbf{X}^a\mathbf{Z}^b.
\]

Define the phased Bell states as
\begin{equation*}
\begin{split}
\ket{\Psi^{(\tau)}_{-a,-b}}_{S_iN_i}
&=
\tau^{ab}\ket{\Psi_{-a,-b}}_{S_iN_i}\\
&=
\left(
W^{(S_i)}(a,b)^\dagger\otimes\mathbb{I}_{N_i}
\right)
\ket{\Phi}_{S_iN_i},
\end{split}
\end{equation*}
where
\[
\ket{\Phi}_{S_iN_i}
=
\frac{1}{\sqrt d}
\sum_{j_i=0}^{d-1}
\ket{j_i,j_i}_{S_iN_i}.
\]
The scalar phase does not change the corresponding projector, and
the phased Bell states satisfy
\[
\left\langle
\Psi^{(\tau)}_{-c,-e}
\middle|
\Psi^{(\tau)}_{-a,-b}
\right\rangle
=
\delta_{c,a}\delta_{e,b}.
\]

After encryption, the state is
\begin{equation*}
\begin{split}
&\ket{\Psi_{enc}}
=
U_{enc}^{(d,n)}\ket{\Psi_0}\\
&=
\frac{1}{d}
\sum_{a,b\in\mathbb{Z}_d}
\varphi(a,b)^{-1}
W^{(A)}(a,b)\ket{\psi}_A
\otimes
\bigotimes_{i=1}^{n}
\ket{\Psi^{(\tau)}_{-a,-b}}_{S_iN_i},
\end{split}
\end{equation*}
where
\[
\ket{\Psi_0}
=
\ket{\psi}_A
\otimes
\bigotimes_{i=1}^{n}
\ket{\Phi}_{S_iN_i}.
\]

Suppose that the noise qudit $N_2$ is lost. The state on the surviving
systems is
\[
\rho_{er}
=
\Tr_{N_2}
\left(
\ket{\Psi_{enc}}\bra{\Psi_{enc}}
\right).
\]
Replacing the missing action on $N_2$ by an action on its surviving
partner $S_2$ i.e. $\left(W^{(N_2)}(c,e)\right)^{\mathsf T}
\longrightarrow
W^{(S_2)}(c,e)$, the modified decryption operator for recovering the
input state on $S_1$ is
\begin{equation*}
\begin{split}
\tilde{U}_{dec}^{(d,n)}
&=
\sum_{c,e\in\mathbb{Z}_d}
\varphi(c,e)
\ket{\Psi_{-c,-e}}
\bra{\Psi_{-c,-e}}_{S_1N_1}\\
&\quad\otimes W^{(S_2)}(c,e)
\otimes
\bigotimes_{i=3}^{n}
\left(W^{(N_i)}(c,e)\right)^{\mathsf T}.
\end{split}
\end{equation*}

Since $\tilde{U}_{dec}^{(d,n)}$ does not act on $N_2$, it may be
applied before $N_2$ is traced out:
\begin{equation*}
\begin{split}
&\tilde{U}_{dec}^{(d,n)}
\rho_{er}
\tilde{U}_{dec}^{(d,n)\dagger}\\
&=
\Tr_{N_2}
\left(
\ket{\Omega}\bra{\Omega}
\right),
\end{split}
\end{equation*}
where
\[
\ket{\Omega}
=
\left(
\tilde{U}_{dec}^{(d,n)}
\otimes\mathbb{I}_{N_2}
\right)
\ket{\Psi_{enc}}.
\]

For each term in $\ket{\Psi_{enc}}$, Bell orthogonality gives
\[
\ket{\Psi_{-c,-e}}
\bra{\Psi_{-c,-e}}
\ket{\Psi^{(\tau)}_{-a,-b}}
=
\delta_{c,a}\delta_{e,b}
\ket{\Psi^{(\tau)}_{-a,-b}}.
\]
Thus, the sums over $(c,e)$ collapse to $(c,e)=(a,b)$. When the
density operator is formed, this selection occurs independently on
the ket and bra terms, so no condition
$(a,b)=(a',b')$ is imposed.

The corresponding Weyl operators satisfy
\begin{equation*}
\begin{split}
\left(
W^{(S_2)}(a,b)\otimes\mathbb{I}_{N_2}
\right)
\ket{\Psi^{(\tau)}_{-a,-b}}_{S_2N_2}
&=
\ket{\Phi}_{S_2N_2},\\
\left(
\mathbb{I}_{S_i}\otimes
W^{(N_i)}(a,b)^{\mathsf T}
\right)
\ket{\Psi^{(\tau)}_{-a,-b}}_{S_iN_i}
&=
\ket{\Phi}_{S_iN_i},
\quad i\geq3.
\end{split}
\end{equation*}
The factors $\varphi(a,b)$ and $\varphi(a,b)^{-1}$ also cancel.
Therefore,
\begin{equation*}
\begin{split}
\ket{\Omega}
&=
\frac{1}{d}
\sum_{a,b\in\mathbb{Z}_d}
W^{(A)}(a,b)\ket{\psi}_A
\otimes
\ket{\Psi^{(\tau)}_{-a,-b}}_{S_1N_1}\\
&\quad\otimes
\ket{\Phi}_{S_2N_2}
\otimes
\bigotimes_{i=3}^{n}
\ket{\Phi}_{S_iN_i}.
\end{split}
\end{equation*}

To evaluate the remaining sum, write
\[
\ket{\psi}_A
=
\sum_{k=0}^{d-1}\psi_k\ket{k}_A.
\]
Using $\tau^2=\omega$ and
\[
\sum_{b\in\mathbb{Z}_d}
\omega^{b(k-j+a)}
=
d\,\delta_{j,k+a},
\]
we obtain
\begin{equation*}
\begin{split}
&\frac{1}{d}
\sum_{a,b}
W^{(A)}(a,b)\ket{\psi}_A
\otimes
\ket{\Psi^{(\tau)}_{-a,-b}}_{S_1N_1}\\
&=
\frac{1}{d\sqrt d}
\sum_{a,b,k,j}
\psi_k\,
\omega^{b(k-j+a)}
\ket{k+a}_A
\ket{j-a}_{S_1}
\ket{j}_{N_1}\\
&=
\frac{1}{\sqrt d}
\sum_{a,k}
\psi_k
\ket{k+a}_A
\ket{k}_{S_1}
\ket{k+a}_{N_1}\\
&=
\ket{\Phi}_{AN_1}\otimes\ket{\psi}_{S_1}.
\end{split}
\end{equation*}
Hence,
\[
\ket{\Omega}
=
\ket{\psi}_{S_1}
\otimes
\ket{\Phi}_{AN_1}
\otimes
\ket{\Phi}_{S_2N_2}
\otimes
\bigotimes_{i=3}^{n}
\ket{\Phi}_{S_iN_i}.
\]

Finally, tracing out $N_2$ gives
\begin{equation*}
\begin{split}
&\tilde{U}_{dec}^{(d,n)}
\rho_{er}
\tilde{U}_{dec}^{(d,n)\dagger}\\
&=
\ket{\psi}\bra{\psi}_{S_1}
\otimes
\ket{\Phi}\bra{\Phi}_{AN_1}
\otimes
\frac{\mathbb{I}_{S_2}}{d}
\otimes
\bigotimes_{i=3}^{n}
\ket{\Phi}\bra{\Phi}_{S_iN_i}.
\end{split}
\end{equation*}
Therefore,
\[
\Tr_{AS_2S_3,\cdots,N_1,N_3,\cdots,N_n}
\left(
\tilde{U}_{dec}^{(d,n)}
\rho_{er}
\tilde{U}_{dec}^{(d,n)\dagger}
\right)
=
\ket{\psi}\bra{\psi}_{S_1},
\]
Hence the original input state $\ket{\psi}_A$ is recovered on $S_1$ despite
the loss of $N_2$, by replacing the missing noise qudit with the
signal qudit $S_2$.

\subsection*{Appendix F: SWAP operator for generalised dimension}

Let
\[
\mathcal{S}=\frac{1}{d}\sum_{a,b\in \mathbb{Z}_d} W(a,b)\otimes W(a,b)^\dagger .
\]
We check its action on \(\ket{i}\otimes \ket{j}\).

Using
\[
\Xop^a \Zop^b \ket{i}=\omega^{bi}\ket{i+a},\qquad \omega=e^{2\pi i/d},
\]
we get
\[
W(a,b)\ket{i}=\tau^{ab}\omega^{bi}\ket{i+a}.
\]

Also
\[
W(a,b)^\dagger \ket{j}
=\tau^{-ab}\omega^{-b(j-a)}\ket{j-a},
\]
so
\[
\bigl(W(a,b)\otimes W(a,b)^\dagger\bigr)\ket{i,j}
=\omega^{b(i-j+a)}\ket{i+a,j-a}.
\]

Therefore
\[
\mathcal{S}\ket{i,j}
=\frac{1}{d}\sum_{a,b}\omega^{b(i-j+a)}\ket{i+a,j-a}.
\]

Now use
\[
\sum_{b\in \mathbb{Z}_d}\omega^{bm}
=d\,\delta_{m,0 \!\!\!\pmod d}.
\]
So only the term with
\[
i-j+a\equiv 0 \pmod d
\]
survives, i.e.
\[
a\equiv j-i.
\]

Hence
\[
\mathcal{S}\ket{i,j}=\ket{j,i}.
\]
So indeed
\[
\mathcal{S}=\mathrm{SWAP}.
\].

\subsection*{Appendix G: Additional derivation for proof of Theorem 2 ($n = 2$)}

Let us compute the marginals of each group by further breaking them into subgroups. By symmetry, we will only need to compute 6 marginals.
\paragraph{Group 1: $\rho_{S_1N_1}$ (Within-pair)}
To find $\rho_{S_1N_1}$, we trace over the registers $\{A, S_2, N_2\}$. The partial trace imposes the following delta constraints:
\begin{itemize}
    \item Trace over $A$: $k+a = k'+a' \implies k' = k + (a-a')$. Let $c = a-a'$.
    \item Trace over $S_2$: $j_2-a = j_2'-a' \implies j_2' = j_2-c$.
    \item Trace over $N_2$: $j_2 = j_2'$.
\end{itemize}
Substituting $j_2' = j_2$ into $j_2' = j_2-c$ forces $c=0$, which implies $a=a'$ and $k=k'$. 
With $a=a'$, the phase $\tau^{a^2-a'^2} = 1$. The sum over $b, b'$ decouples from the spatial indices:
\begin{equation*}
\begin{split}
\rho_{S_1N_1}& = \frac{1}{d^4} \sum_{a,b,b',k,j_1,j_1',j_2} \tau^{b^2-b'^2} |\psi_k|^2 \omega^{b(k-j_1-j_2)}\\
&\omega^{-b'(k-j_1'-j_2)} |j_1-a, j_1\rangle \langle j_1'-a, j_1'| \\
&= \frac{1}{d^4} \sum_{a,j_1,j_1'}\left( \sum_k |\psi_k|^2 \right) \\
&\left( \sum_{b,b',j_2} \tau^{b^2-b'^2} \omega^{b(k-j_1-j_2)-b'(k-j_1'-j_2)} \right)\\
&\times|j_1-a, j_1\rangle \langle j_1'-a, j_1'|
\end{split}
\end{equation*}
Let $S$ represent the sum over the indices traced over or summed out during marginalisation ($j_2, b, b'$):
\[
S = \sum_{j_2=0}^{d-1} \sum_{b, b'=0}^{d-1} \tau^{b^2-b'^2} \omega^{b(k-j_1-j_2)-b'(k-j_1'-j_2)}
\]
Rearranging the exponents to group the terms by $j_2$:
\[
S = \sum_{b, b'=0}^{d-1} \tau^{b^2-b'^2} \omega^{b(k-j_1) - b'(k-j_1')} \left( \sum_{j_2=0}^{d-1} \omega^{(b'-b)j_2} \right)
\]

Using the orthogonality identity $\sum_{j_2=0}^{d-1} \omega^{(b'-b)j_2} = d\,\delta_{b,b'}$:
\[
S = \sum_{b, b'=0}^{d-1} \tau^{b^2-b'^2} \omega^{b(k-j_1) - b'(k-j_1')} \cdot (d\delta_{b,b'})
\]
The Kronecker delta $\delta_{b,b'}$ forces $b = b'$, leading to two critical simplifications:
\begin{enumerate}
    \item The quadratic phase becomes $\tau^{b^2-b^2} = \tau^0 = 1$.
    \item The signal index $k$ in the $\omega$ terms cancels: $\omega^{bk - bk} = 1$.
\end{enumerate}
The expression simplifies to:
\[
S = d \sum_{b=0}^{d-1} \omega^{b(j'_1 - j_1)}
\]

Applying the orthogonality identity again to the remaining sum over $b$:
\[
\sum_{b=0}^{d-1} \omega^{b(j'_1 - j_1)} = d\delta_{j_1, j_1'}
\]
Thus, the total internal sum collapses to:
\[
S = d \cdot d\delta_{j_1, j_1'} = d^2\delta_{j_1, j_1'}
\]

Substituting $S$ back into the marginal density matrix $\rho_{S_1N_1}$ and since $\sum_k |\psi_k|^2 = 1$:
\[
\rho_{S_1N_1} = \frac{1}{d^4} \sum_{a, j_1, j_1'} (d^2\delta_{j_1, j_1'}) |j_1-a, j_1\rangle \langle j_1'-a, j_1'|
\]
Applying the delta $\delta_{j_1, j_1'}$:
\[
\rho_{S_1N_1} = \frac{d^2}{d^4} \sum_{a=0}^{d-1} \sum_{j_1=0}^{d-1} |j_1-a, j_1\rangle \langle j_1-a, j_1|
\]
The double sum over $a$ and $j_1$ spans all $d^2$ basis states $|m, n\rangle$ in the $\mathbb{Z}_d \otimes \mathbb{Z}_d$ Hilbert space. Specifically, for each fixed $j_1$, as $a$ cycles through $\{0,1,\cdots ,d-1\}$, the first index $m = j_1-a \pmod d$ also cycles through $\{0,1,\cdots, d-1\}$. This yields the normalized identity matrix:
\[
\rho_{S_1N_1} = \frac{1}{d^2} \sum_{m=0}^{d-1} \sum_{n=0}^{d-1} |m, n\rangle \langle m, n| = \frac{1}{d^2} \mathbb{I}_{d^2}
\]
One can easily check $\rho_{S_2N_2}$ using the same argument.

\paragraph{Group 2a: $\rho_{AN_1}$ ($A$ + one noise qudit)}
To find $\rho_{AN_1}$, we trace over $\{S_1, S_2, N_2\}$. The partial trace yields the following Kronecker delta constraints from the traced indices:
\begin{itemize}
    \item Trace over $S_1$: $\delta_{j_1-a, j_1'-a'} \implies j_1' = j_1 - (a-a')$. Let $c = a-a'$.
    \item Trace over $S_2$: $\delta_{j_2-a, j_2'-a'} \implies j_2' = j_2 - c$.
    \item Trace over $N_2$: $\delta_{j_2, j_2'}$.
\end{itemize}
Combining $\delta_{j_2, j_2'}$ with $j_2' = j_2-c$ forces $c=0$. This implies:
\[ c=0 \implies a = a' \quad \text{and} \quad j_1' = j_1 \]
Substituting these back into the density matrix, the phase $\tau^{a^2-a'^2}$ becomes $\tau^0 = 1$. The sum becomes:
\begin{equation*}
\begin{split}
\rho_{AN_1} &= \frac{1}{d^4} \sum_{a,j_1,k,k',b,b',j_2} \tau^{b^2-b'^2} \psi_k \psi_{k'}^* \omega^{b(k-j_1-j_2)}\\
&\omega^{-b'(k'-j_1-j_2)} |k+a, j_1\rangle \langle k'+a, j_1|
\end{split}
\end{equation*}
Isolating the sum over $j_2$ to collapse $b, b'$:
\[ \sum_{j_2=0}^{d-1} \omega^{(b'-b)j_2} = d\delta_{b,b'} \]
With $b=b'$, the phase $\tau^{b^2-b'^2} = 1$. Now collapse $k, k'$ by summing over $b$:
\[ \sum_{b=0}^{d-1} \omega^{b(k-k')} = d\delta_{k,k'} \]
Under the condition $|\psi_k|^2 = 1/d$:
\begin{align*}
\rho_{AN_1} &= \frac{d^2}{d^4} \sum_{a,j_1,k} \frac{1}{d} |k+a, j_1\rangle \langle k+a, j_1| \\
&= \frac{1}{d^3} \sum_{a,k,j_1} |k+a, j_1\rangle \langle k+a, j_1|
\end{align*}
Since $m = k+a \pmod d$ spans $\mathbb{Z}_d$ for any $a$, and there are d choices of $a$:
\[ \rho_{AN_1} = \frac{d}{d^4} \cdot d \sum_{m,j_1} |m, j_1\rangle \langle m, j_1| = \frac{d^2}{d^4} \mathbb{I}_{d^2} = \frac{1}{d^2} \mathbb{I}_{d^2} \]

One can easily show this for $\rho_{AN_2}$.

\paragraph{Group 2b: $\rho_{AS_1}$ ($A$ + one signal qudit)}

To find $\rho_{AS_1}$, we trace over $\{S_2, N_1, N_2\}$. The trace constraints are:
\begin{itemize}
    \item Trace over $N_1$: $\delta_{j_1, j_1'}$.
    \item Trace over $N_2$: $\delta_{j_2, j_2'}$.
    \item Trace over $S_2$: $\delta_{j_2-a, j_2'-a'} \implies j_2-a = j_2-a' \implies a = a'$.
\end{itemize}
Just as before, $a=a'$ implies $c=0$. Since $\delta_{j_1, j_1'}$ is already given by the trace over $N_1$, we have $j_1 = j_1'$ and $a=a'$. The resulting sum is:
\begin{equation*}
\begin{split}
\rho_{AS_1} &= \frac{1}{d^4} \sum_{a,j_1,k,k',b,b',j_2} \tau^{b^2-b'^2} \psi_k \psi_{k'}^* \omega^{b(k-j_1-j_2)}\\
&\omega^{-b'(k'-j_1-j_2)} |k+a, j_1-a\rangle \langle k'+a, j_1-a|
\end{split}
\end{equation*}
The collapse of $b, b'$ (via $j_2$ sum) and $k, k'$ (via $b$ sum) gives:
\[ \rho_{AS_1} = \frac{d^2}{d^4} \sum_{a,k,j_1} |\psi_k|^2 |k+a, j_1-a\rangle \langle k+a, j_1-a| \]
Let $m = k+a$ and $n = j_1-a$. For a fixed $a$, the pair $(m, n)$ spans all $d^2$ basis states. Summing over $a$ gives a factor of d:
\[ \rho_{AS_1} = \frac{d^2}{d^4} \cdot \frac{1}{d} \cdot d \sum_{m,n} |m, n\rangle \langle m, n| = \frac{1}{d^2} \mathbb{I}_{d^2} \]

One can easily show for $\rho_{AS_2}$.

\paragraph{Group 3a: $\rho_{S_1S_2}$ (signal-signal pair)}

To compute $\rho_{S_1S_2}$, we trace over registers $\{A,N_1,N_2\}$.

The trace imposes the following constraints:

\begin{itemize}
\item Trace over $N_1$: $\delta_{j_1,j'_1}$
\item Trace over $N_2$: $\delta_{j_2,j'_2}$
\item Trace over $A$: $k+a = k'+a'$
\end{itemize}

Let
\[
c = a-a', \qquad k' = k+c, \qquad a' = a-c.
\]

The marginal becomes

\begin{equation*}
\begin{split}
\rho_{S_1S_2}
&=
\frac{1}{d^4}
\sum_{a,c,b,b',k,j_1,j_2}
\tau^{a^2-(a-c)^2+b^2-b'^2}\\
&\times \psi_k \psi_{k+c}^*
\omega^{b(k-j_1-j_2)}\omega^{-b'(k+c-j_1-j_2)}\\
&\times
|j_1-a,j_2-a\rangle
\langle j_1-a+c,j_2-a+c|.
\end{split}
\end{equation*}

The quadratic phase simplifies:

\[
a^2-(a-c)^2 = 2ac-c^2
\]

so

\[
\tau^{a^2-(a-c)^2} = \tau^{2ac-c^2}.
\]

Assuming uniform input $|\psi_k|^2 = 1/d$,

\[
\frac{1}{d}
\sum_k
\omega^{(b-b')k}
=
\delta_{b,b'}.
\]

This consumes the $b'$ sum and removes the quadratic phase $\tau^{b^2-b'^2}$.

The density matrix becomes

\begin{equation*}
\begin{split}
\rho_{S_1S_2}
&=
\frac{1}{d^4}
\sum_{a,c,b,j_1,j_2}
\tau^{2ac-c^2}
\omega^{-bc}\\
&\times|j_1-a,j_2-a\rangle
\langle j_1-a+c,j_2-a+c|.
\end{split}
\end{equation*}

Now sum over $b$:

\[
\sum_{b=0}^{d-1}\omega^{-bc} = d\delta_{c,0}.
\]

This forces

\[
c=0 \Rightarrow a=a'.
\]

Substituting $c=0$ gives

\[
\rho_{S_1S_2}
=
\frac{d}{d^4}
\sum_{a,j_1,j_2}
|j_1-a,j_2-a\rangle
\langle j_1-a,j_2-a|.
\]

Let

\[
m=j_1-a,\qquad n=j_2-a.
\]

As $a,j_1,j_2$ run over $\mathbb{Z}_d$, the pair $(m,n)$ spans all basis states of $\mathbb{Z}_d\otimes\mathbb{Z}_d$ exactly d times.

Thus

\[
\sum_{a,j_1,j_2}
|j_1-a,j_2-a\rangle
\langle j_1-a,j_2-a|
=
d \mathbb{I}_{d^2}.
\]

Therefore

\[
\rho_{S_1S_2}
=
\frac{d}{d^4}\cdot d \mathbb{I}_{d^2}
=
\frac{d^2}{d^4}\mathbb{I}_{d^2}
=
\frac{1}{d^2}\mathbb{I}_{d^2}.
\]

Hence the signal-signal marginal is maximally mixed.

\paragraph{Group 3b: $\rho_{N_1N_2}$ (noise-noise pair)}

Tracing over $\{A,S_1,S_2\}$ gives $k'=k+c$, $j'_1=j_1-c$ and $j'_2=j_2-c$ respectively. The marginal becomes

\begin{equation*}
\begin{split}
\rho_{N_1N_2}
&=
\frac{1}{d^4}
\sum_{a,c,b,b',k,j_1,j_2}
\tau^{2ac-c^2+b^2-b'^2}
\psi_k\psi_{k+c}^*
\omega^{b(k-j_1-j_2)}\\
&\omega^{-b'(k-j_1-j_2+3c)}
|j_1,j_2\rangle
\langle j_1-c,j_2-c|.
\end{split}
\end{equation*}

Uniform input gives

\[
\frac{1}{d}\sum_k \omega^{(b-b')k}
=
\delta_{b,b'}.
\]

Thus

\[
\rho_{N_1N_2}
=
\frac{1}{d^4}
\sum_{a,c,b,j_1,j_2}
\tau^{2ac-c^2}\omega^{-3bc}
|j_1,j_2\rangle
\langle j_1-c,j_2-c|.
\]

Summing over $a$ gives factor d:
\[
\sum_a\tau^{2ac}=d\delta_{c,0}, \,\ 
\]
Therefore,

\[
\rho_{N_1N_2}
=
\frac{d}{d^4}
\sum_{c,b,j_1,j_2}
\tau^{-c^2}\omega^{-3bc}\delta_{c,0}
|j_1,j_2\rangle
\langle j_1-c,j_2-c|
\]

Now sum over $b$ gives a factor of $d$:

Thus

\[
\rho_{N_1N_2}
=
\frac{d}{d^4}\cdot d
\sum_{j_1,j_2}
|j_1,j_2\rangle\langle j_1,j_2|
=
\frac{d^2}{d^4}\mathbb{I}_{d^2}
=
\frac{1}{d^2}\mathbb{I}_{d^2}.
\]

\paragraph{Group 3c: $\rho_{S_1N_2}$ (cross-pair)}

Tracing over $\{A,S_2,N_1\}$ gives

\[
j'_2=j_2-c,\quad k'=k+c, \quad j_1=j'_1.
\]

The marginal becomes

\begin{equation*}
\begin{split}
\rho_{S_1N_2}
&=
\frac{1}{d^4}
\sum_{a,c,b,b',k,j_1,j_2}
\tau^{2ac-c^2+b^2-b'^2}\\
&\times \psi_k\psi_{k+c}^*
\omega^{b(k-j_1-j_2)}\omega^{-b'(k-j_1-j_2+2c)}\\
&\times 
|j_1-a,j_2\rangle
\langle j_1-a+c,j_2-c|.
\end{split}
\end{equation*}

Uniform input gives

\[
\frac{1}{d}\sum_k \omega^{(b-b')k}
=
\delta_{b,b'}.
\]

Thus

\begin{equation*}
\begin{split}
\rho_{S_1N_2}
&=
\frac{1}{d^4}
\sum_{a,c,b,j_1,j_2}
\tau^{2ac-c^2}
\omega^{-2bc}\\
&|j_1-a,j_2\rangle
\langle j_1-a+c,j_2-c|.
\end{split}
\end{equation*}

Summing over $b$:

\[
\sum_b \omega^{-2bc} = d\delta_{c,0}, \,\ \text{gcd}(2,d)=1
\]
when $d$ is odd.

This simplifies the sum and yields,

\[
\rho_{S_1N_2}
=
\frac{d}{d^4}
\sum_{a,j_1,j_2}
|j_1-a,j_2\rangle
\langle j_1-a,j_2|.
\]

Let

\[
m=j_1-a,\qquad n=j_2.
\]

This spans all $d^2$ basis states $d$ times:

\[
\sum_{a,j_1,j_2}
|j_1-a,j_2\rangle\langle j_1-a,j_2|
=
d\mathbb{I}_{d^2}.
\]

Thus

\[
\rho_{S_1N_2}
=
\frac{d}{d^4}\cdot d \mathbb{I}_{d^2}
=
\frac{1}{d^2}\mathbb{I}_{d^2}.
\]

When $d$ is even, we have

\[
\sum_{b=0}^{d-1} \omega^{-2bc} = d \left( \delta_{c,0} + \delta_{c, \frac{d}{2}} \right), \,\ \text{gcd}(2,d)=2.
\]
Plugging this back into the expression for $\rho_{S_1 N_2}$, the density matrix splits into two components: $\rho_{S_1 N_2} = \frac{1}{d^3} \left( \Sigma_{c=0} + \Sigma_{c=\frac{d}{2}} \right)$.
For $c=0$, the term behaves exactly as it does in the odd case:$$\Sigma_{c=0} = \sum_{a, j_1, j_2} \tau^{0} |j_1 - a, j_2\rangle \langle j_1 - a, j_2|$$Letting $m = j_1 - a$, as $j_1$ and $a$ run from $0$ to $d-1$, the index $m$ spans all $d$ basis states, and the sum over $a$ simply provides a redundancy factor of $d$. This gives:$$\Sigma_{c=0} = d \sum_{m, j_2} |m, j_2\rangle \langle m, j_2| = d \cdot \mathbb{I}_{d^2}$$
and hence $\rho_{S_1N_1}=\frac{1}{d^2}\mathbb{I}_{d^2}$. Now let us look at the extra term introduced when $d$ is even:$$\Sigma_{c=\frac{d}{2}} = \sum_{a, j_1, j_2} \tau^{2a(\frac{d}{2}) - (\frac{d}{2})^2} |j_1 - a, j_2\rangle \langle j_1 - a + \frac{d}{2}, j_2 - \frac{d}{2}|$$Let's simplify the $\tau$ phase factor first:$$\tau^{ad} = \left(e^{\frac{i\pi(d+1)}{d}}\right)^{ad} = e^{i\pi(d+1)a} = (-1)^{a}$$
as $d+1$ is odd. Therefore, $$\Sigma_{c=\frac{d}{2}} = \tau^{-d^2/4} \sum_{a, j_1, j_2} (-1)^a |j_1 - a, j_2\rangle \langle j_1 - a + \frac{d}{2}, j_2 - \frac{d}{2}|$$

 We then perform a change of variables where $m = j_1 - a$. For any fixed value of $a$ and $j_2$, as $j_1$ shifts from $0$ to $d-1$, $m$ still cycles through the complete set of basis states exactly once. Rewriting the sum using $m$:

 \begin{equation*}
 \begin{split}
 \Sigma_{c=\frac{d}{2}} &= \tau^{-\frac{d^2}{4}} \sum_{j_2=0}^{d-1} \underbrace{\sum_{a=0}^{d-1} (-1)^a}_{=0 \text{ ($d$ is even)}}\\ &\times \left( \sum_{m=0}^{d-1} |m, j_2\rangle \langle m + \frac{d}{2}, j_2 - \frac{d}{2}| \right)
 \end{split}
 \end{equation*}

Combining both parts back into your marginal expression:$$\rho_{S_1 N_2} = \frac{1}{d^3} \left( d \cdot \mathbb{I}_{d^2} + 0 \right) = \frac{1}{d^2} \mathbb{I}_{d^2}$$. Hence, both even and odd dimensions obtain the exact same maximally mixed state. By symmetry, one can also easily see that $\rho_{N1S_2}$ is also maximally mixed.


\subsection*{Appendix H: Theorem 2 for $n>2$ and PEME for three qubit pairs}

\subsubsection{Proof of Theorem 2 for $n>2$}
\begin{proof}
It is sufficient to show that a particular marginal state of $n=3$, is not maximally mixed. Let us compute the marginals $\rho_{AS_1N_1}$. Here, we trace over $\{S_2,N_2,S_3,N_3\}$ and this gives:

\begin{itemize}
    \item Trace over $N_2$: $j_2=j'_2$
    \item Trace over $N_3$: $j_3=j'_3$
    \item Trace over $S_2$: $j'_2=j_2-c$
    \item Trace over $S_3$: $j'_3=j_3-c$
\end{itemize}
where $c=a-a'$. This means $c=0$. We take $\psi_k$ to be uniform, hence $\psi_k\psi_{k'}^*=|\psi_k|^2=1/d$. The marginals then becomes:

\begin{equation*}
\begin{split}
&\rho_{AS_1N_1} = \frac{1}{d^5} \sum_{a,b,b',k,k',j_1,j_1',j_2,j_3} \tau^{b^2-b'^2} |\psi_k|^2 \omega^{b(k-j_1-j_2-j_3)}\\
&\times \omega^{-b'(k'-j_1'-j_2-j_3)}|k+a,j_1-a, j_1\rangle \langle k'+a,j_1'-a, j_1'|
\end{split}
\end{equation*}

Since $j_2,j_3$ are free variables not in the braket, we can sum over $j_2,j_3$ to enforce $b=b'$.

\[
\sum_{j_2,j_3}\omega^{(j_2+j_3)(b'-b)}= d^2\delta_{b,b'}
\]

This simplifies $\rho_{AS_1N_1}$ to:
\begin{equation*}
\begin{split}
\rho_{AS_1N_1} &= \frac{1}{d^3} \sum_{a,b,k,k',j_1,j_1'}|\psi_k|^2 \omega^{b(k-k'-j_1+j'_1)}\\
&|k+a,j_1-a, j_1\rangle \langle k'+a,j_1'-a, j_1'|.
\end{split}
\end{equation*}

Summing over $b$ yields:

\begin{equation*}
\sum_b \omega^{b(k-k'-j_1+j'_1)}=d\,\delta_{k-k',j_1-j'_1}
\end{equation*}

Unlike the AME$(5,d)$ case where we do not have such marginals with $A$ and same pairings, this is a single constraint on a four variable system ($k, j_1, k', j'_1$). It only requires the sum of the indices to match, which allows for off-diagonal terms where $k\neq k'$, $j'_1\neq j'_1$ but $k-j_1=k'-j'_1$ which hinders it from becoming a maximally mixed state. To see this, let $d=2$:

\begin{equation*}
\begin{split}
\rho_{AS_1N_1}&=\frac{1}{8}\sum_{a,k,k',j_1,j'_1\in Z_2}\delta_{k-k',j_1-j'_1}\\
&\times|k+a,j_1-a, j_1\rangle \langle k'+a,j_1'-a, j_1'|.
\end{split}
\end{equation*}
one can easily check that this is in fact equal to:

\begin{equation*}
    \rho_{AS_1N_1} = \frac{1}{8} (\mathbb{I}_8 + X \otimes X \otimes X)
\end{equation*}
which is clearly not a maximally mixed state.

\end{proof}

\subsubsection{Three qudit pairs: marginal counts and the PEME condition}

The exceptional marginal above rules out AME. To test the PEME
condition, we must determine how many reductions are maximally
mixed at each size and compare these counts with the largest
possible counts in the same local dimension.
We take $n=3$ and $d\geq2$ throughout.

First retain the reference $A_1$ as in Section~V. The four initial
Bell pairs join corresponding entries of
\[
 L=(A_1,N_1,N_2,N_3),\qquad B=(A_2,S_1,S_2,S_3).
\]
In each term of the encoder, the Weyl label $(a,b)$ determines the
Bell state on every pair. Tracing out a complete pair therefore
enforces $a=a'$ and $b=b'$, leaving a uniform mixture over the
$d^2$ labels.

Consider a four qudit subset $E$. If $E$ contains two complete
pairs, its reduction has rank $d^2$ in dimension $d^4$. If $E$
contains one complete pair and one half of each of two other pairs,
the Bell label averaging gives $\mathbb I_{d^2}/d^2$ on the complete
pair. Each retained half gives $\mathbb I_d/d$ for every label,
so $\rho_E=\mathbb I_{d^4}/d^4$. There are $4\cdot3\cdot2^2=48$
such subsets, choosing the complete pair, the omitted pair, and the
two retained halves. This part of the count is independent of $d$.

\smallskip
\noindent\textit{All four pairs split.}
Suppose $E$ retains one member of each pair. Let $t$ be the
number of retained registers from $L=(A_1,N_1,N_2,N_3)$.
We compute $\rho_E$ directly in the computational basis. As in
Appendix~D, expanding the encoder gives
\[
\begin{aligned}
 \ket{\Psi'_{enc}}
 &=\frac{1}{d^3}\sum_{a,b,k,j_1,j_2,j_3}
   \tau^{a^2+b^2}\omega^{b(k-j_1-j_2-j_3)}\\
 &\quad{}\times
   \ket{k,k+a}_{A_1A_2}
   \bigotimes_{i=1}^3\ket{j_i-a,j_i}_{S_iN_i},
\end{aligned}
\]
where all indices are in $\mathbb Z_d$.

First take $A_1\in E$. Then $E$ contains $A_1$, $t-1$ noise
qudits and $4-t$ signal qudits. Writing $c=a-a'$ as above,
the partial trace gives
\begin{itemize}
 \item Trace over $A_2$: $k+a=k'+a'$, hence $k'=k+c$.
 \item If $N_i$ is retained, trace over $S_i$:
       $j_i-a=j_i'-a'$, hence $j_i'=j_i-c$.
 \item If $S_i$ is retained, trace over $N_i$: $j_i'=j_i$.
\end{itemize}
Since $t-1$ noise qudits are retained, it follows that
\[
 k'-\sum_{i=1}^3j_i'=k-\sum_{i=1}^3j_i+tc.
\]
Thus, at fixed $b-b'$, the sum over $b'$ gives
\[
 \sum_{b'\in\mathbb Z_d}\omega^{b'[(b-b')-tc]}
 =d\,\delta_{b-b',\,tc}.
\]
Here we used $\tau^2=\omega$ and
$b^2-b'^2=(b-b')^2+2b'(b-b')$.

Now fix the ket of a matrix element of $\rho_E$. For each $c$,
the trace conditions above determine its bra. Substituting
$a'=a-c$ and $b-b'=tc$, this matrix element becomes
\[
 \frac{\tau^{(t^2-1)c^2}}{d^5}
 \sum_{a\in\mathbb Z_d}
 \omega^{ac+tc(k-j_1-j_2-j_3)}.
\]
To evaluate the sum, recall which indices are fixed by the
retained ket. Register $A_1$ fixes $k$, and each retained $N_i$
fixes $j_i$. Each retained $S_i$, however, fixes $j_i-a$.
Increasing $a$ by one therefore increases $j_i$ by one for each
of the $4-t$ retained signals. The sum is consequently a
geometric sum with ratio $\omega^{[1-t(4-t)]c}$. It survives
precisely when
\[
 [1-t(4-t)]c=0\pmod d.
\]

For $c=0$, the bra equals the ket and every diagonal entry is
$1/d^4$. Each nonzero solution gives a nonzero off-diagonal entry,
since the $A_1$ labels are $k$ and $k+c$. Hence
\[
 \rho_E=\frac{\mathbb I_{d^4}}{d^4}
 \quad\Longleftrightarrow\quad
 \gcd\bigl(d,|1-t(4-t)|\bigr)=1.
\]
If $A_1$ belongs to $E^c$, we apply the same calculation to $E^c$.
The two reductions have the same spectrum because the full state
is pure, and replacing $t$ by $4-t$ leaves the criterion unchanged.

There are $\binom4t$ choices of $E$ for each $t$. For $t=0,4$,
the criterion always holds, giving two subsets. For $t=1,3$,
it holds when $d$ is odd, giving eight more. For $t=2$, it holds
when $d$ is not divisible by three, giving six more. Since $|1-t(4-t)|$ takes only the values $1$, $2$, and $3$, the
count depends only on whether $d$ is odd or even and whether it is
divisible by three. Thus $\gcd(6,d)$ records both tests. For odd $d$, it equals $1$ if $d$ is not divisible by three and $3$ if it is.
For even $d$, it equals $2$ if $d$ is not divisible by three and
$6$ if it is. In particular,
for $d=2$ the allowed values are $t=0,2,4$, so exactly
$1+6+1=8$ all-split reductions are maximally mixed.

\smallskip
\noindent\textit{Uniform input.}
Projecting $A_1$ onto the uniform state
$\ket{\psi}=d^{-1/2}\sum_{k=0}^{d-1}\ket{k}$ prepares the same
input on $A_2$, now denoted by $A$. The outcome has probability
$1/d$, so for any output triple $T$ the normalized reduction is
\[
 \rho_T^{(7)}
 =d\,\bra{\psi}_{A_1}\rho_{A_1T}^{(8)}\ket{\psi}_{A_1}.
\]
When $A_1T$ is all split, this projection sums the matrix element
computed above over $k$. Its only $k$-dependent factor is
$\omega^{tck}$, so the sum enforces $tc=0\pmod d$.
The partial trace already requires $c=t(4-t)c\pmod d$.
Together these conditions give $c=0$. Thus all off-diagonal
entries vanish, and each diagonal entry is $d/d^4=1/d^3$.
Every all-split choice therefore gives
$\rho_T^{(7)}=\mathbb I_{d^3}/d^3$ for the uniform input.
The only exceptional triples are therefore
$T=\{A,S_i,N_i\}$, $i=1,2,3$, for which $A_1T$ consists of two
complete pairs. As in the partial trace above, each of these
triples is a uniform mixture of $d^2$ orthogonal rank-one terms
labelled by $(a,b)$, so its rank is $d^2<d^3$. This leaves
$\binom73-3=32$ maximally mixed triples for every $d$.

\begin{center}
\small
\begin{tabular}{@{}ccc@{}}
\toprule
& \textbf{Eight qudit state} & \textbf{Seven qudit state} \\
$\gcd(6,d)$ & Mixed maximally on & Mixed maximally on \\
\midrule
$1$ & $64$ of $70$ quadruples & $32$ of $35$ triples \\
$2$ & $56$ of $70$ quadruples & $32$ of $35$ triples \\
$3$ & $58$ of $70$ quadruples & $32$ of $35$ triples \\
$6$ & $50$ of $70$ quadruples & $32$ of $35$ triples \\
\bottomrule
\end{tabular}
\end{center}

Every subset of at most three of the eight qudits omits a complete
pair and contains at most one complete pair. After tracing out the omitted pair, all $d^2$ Bell labels $(a,b)$
occur with equal probability. The resulting state on each of these
subsets is therefore maximally mixed. Projecting
$A_1$ gives the same property for all one and two qudit reductions
of the seven qudit state. One can also check the results in the table numerically.

At $d=2$, the counts $32$ and $56$ attain the known seven and
eight qubit extrema, respectively.
Both states are therefore PEME. For every $d\geq3$, an
$\mathrm{AME}(7,d)$ state exists and has all 35
triples maximally mixed. Thus the seven qudit uniform-input state
is PEME exactly at $d=2$. For the eight qudit extension, the table also
gives the achieved count in every dimension. As in the seven qudit case, the eight qudit encrypted state fails
to be PEME in any dimension $d$ for which an $\mathrm{AME}(8,d)$
state exists.

\subsection*{Appendix I: Fourier Covariance of the $n=2$ Encryption Unitary}

\begin{proof}
Let $U_{enc}^{(d,2)}$ act on the three registers $A,S_1,S_2$. We first prove the result for
$r=1$; the general result then follows by taking the $r$-th power.

For $n=2$, the encryption unitary is
\begin{equation*}
\begin{split}
U_{enc}^{(d,2)}
&=
\frac{1}{d}
\sum_{a,b\in\mathbb Z_d}
\varphi^{-1}(a,b)\,
W^{(A)}(a,b)\otimes \\
&\bigotimes^2_{i=1}\left(W^{(S_i)}(a,b)\right)^\dagger,
\end{split}
\end{equation*}
where
\begin{equation*}
\varphi^{-1}(a,b)
=
\tau^{a^2+b^2-3ab}.
\end{equation*}
Using
\begin{equation*}
W(a,b)\ket{x}
=
\tau^{ab}\omega^{bx}\ket{x+a}
\end{equation*}
and
\begin{equation*}
W(a,b)^\dagger\ket{y}
=
W(-a,-b)\ket{y}
=
\tau^{ab}\omega^{-by}\ket{y-a},
\end{equation*}
we obtain
\begin{align*}
U_{enc}^{(d,2)}\ket{x,y,z}
&=
\frac{1}{d}
\sum_{a,b\in\mathbb Z_d}
\tau^{a^2+b^2}
\omega^{b(x-y-z)}\\
&\times \ket{x+a,y-a,z-a}.
\label{eq:U_basis_action}
\end{align*}

Now define the local unitary
\begin{equation*}
Q
=
W^{(A)}(1,1)\otimes
W^{(S_1)}(-1,0)\otimes
W^{(S_2)}(-1,0).
\label{eq:local_covariance_operator}
\end{equation*}
Because
\begin{equation*}
W(1,1)\ket{m}
=
\tau\omega^m\ket{m+1},
\quad
W(-1,0)\ket{m}
=
\ket{m-1},
\end{equation*}
applying $Q$ to $U_{enc}^{(d,2)}\ket{x,y,z}$ gives
\begin{align*}
QU_{enc}^{(d,2)}&\ket{x,y,z}
=
\frac{1}{d}
\sum_{a,b\in\mathbb Z_d}
\tau^{a^2+b^2}
\omega^{b(x-y-z)}
\tau\omega^{x+a}
\nonumber\\
&\qquad\qquad\times
\ket{x+a+1,y-a-1,z-a-1}.
\label{eq:QU_before_reindexing}
\end{align*}

Set
\begin{equation*}
a'=a+1,
\qquad
a=a'-1.
\end{equation*}
$QU_{enc}^{(d,2)}\ket{x,y,z}$ becomes
\begin{align*}
QU_{enc}^{(d,2)}&\ket{x,y,z}
=
\frac{1}{d}
\sum_{a',b\in\mathbb Z_d}
\tau^{(a'-1)^2+b^2+1}
\nonumber\\
&\times \omega^{b(x-y-z)}
\omega^{x+a'-1}
\ket{x+a',y-a',z-a'}.
\label{eq:QU_after_reindexing}
\end{align*}
Since $\omega=\tau^2$,
\begin{align*}
\tau^{(a'-1)^2+b^2+1}\omega^{x+a'-1}
&=
\tau^{a'^2-2a'+1+b^2+1}
\tau^{2x+2a'-2}
\nonumber\\
&=
\tau^{a'^2+b^2}\tau^{2x}
\nonumber\\
&=
\omega^x\tau^{a'^2+b^2}.
\end{align*}
Therefore,
\begin{align*}
QU_{enc}^{(d,2)}&\ket{x,y,z}
=
\omega^x
\frac{1}{d}
\sum_{a',b\in\mathbb Z_d}
\tau^{a'^2+b^2}
\omega^{b(x-y-z)}
\\
&\times\ket{x+a',y-a',z-a'}=
\omega^x U_{enc}^{(d,2)}\ket{x,y,z}.
\end{align*}
On the other hand,
\begin{equation*}
\Zop_A\ket{x,y,z}
=
\omega^x\ket{x,y,z}.
\end{equation*}
Hence
\begin{equation*}
QU_{enc}^{(d,2)}\ket{x,y,z}
=
U_{enc}^{(d,2)}\Zop_A\ket{x,y,z}.
\end{equation*}
Because this holds for every computational basis state,
\begin{equation*}
U_{enc}^{(d,2)}\Zop_A=QU_{enc}^{(d,2)}.
\end{equation*}
Multiplying on the right by $U_{enc}^{(d,2)\dagger}$ gives
\begin{equation*}
U_{enc}^{(d,2)}\Zop_AU_{enc}^{(d,2)\dagger}
=
Q
=
W^{(A)}(1,1)\otimes
W^{(S_1)}(-1,0)\otimes
W^{(S_2)}(-1,0).
\label{eq:single_power_covariance}
\end{equation*}

Applying $Q$ repeatedly,
\begin{equation*}
U_{enc}^{(d,2)}\Zop_A^rU_{enc}^{(d,2)\dagger}
=
\left(
U_{enc}^{(d,2)}\Zop_AU_{enc}^{(d,2)\dagger}
\right)^r
=
Q^r.
\end{equation*}
Consequently,
\begin{equation*}
\begin{split}
U_{enc}^{(d,2)}
&\left(
\Zop_A^r\otimes
\mathbb I_{S_1}\otimes
\mathbb I_{S_2}
\right)
U_{enc}^{(d,2)\dagger}\\
&=
\left[
W^{(A)}(1,1)\otimes
W^{(S_1)}(-1,0)\otimes
W^{(S_2)}(-1,0)
\right]^r.
\label{eq:fourier_covariance_full}
\end{split}
\end{equation*}
\end{proof}

\end{document}